\documentclass[10pt,aps,pra,twocolumn,floatfix]{revtex4-2}

\usepackage{preamble}

\begin{document}
\title{Prior-Informed Adaptive Shifts for Sequential Minimal Optimization in Variational Quantum Eigensolvers}

\author{Frederik Stalschus\orcidlink{0009-0007-9196-0814}}
\email{frederik.stalschus@tu-berlin.de}
\affiliation{Berlin Institute for the Foundations of Learning and Data BIFOLD, Berlin, Germany}
\affiliation{Technische Universität Berlin, Berlin, Germany}

\author{Samuele Pedrielli\orcidlink{0009-0005-5973-2235}}
\affiliation{Berlin Institute for the Foundations of Learning and Data BIFOLD, Berlin, Germany}
\affiliation{Technische Universität Berlin, Berlin, Germany}
\affiliation{Padova University, Padova, Italy}

\author{Stefan Kühn\orcidlink{0000-0001-7693-350X}}
\affiliation{Deutsches Elektronen-Synchrotron DESY, Platanenallee 6, 15738 Zeuthen, Germany}

\author{Karl Jansen\orcidlink{0000-0002-1574-7591}}
\affiliation{Deutsches Elektronen-Synchrotron DESY, Platanenallee 6, 15738 Zeuthen, Germany}
\affiliation{Computation-Based Science and Technology Research Center, The Cyprus Institute, Nicosia, Cyprus}

\author{Kim A. Nicoli\orcidlink{0000-0001-5933-1822}}
\affiliation{Oldendorff Carriers GmbH $\&$ Co. KG, Lübeck, Germany}

\author{Shinichi Nakajima\orcidlink{0000-0003-3970-4569}}
\email{nakajima@tu-berlin.de}
\affiliation{Berlin Institute for the Foundations of Learning and Data BIFOLD, Berlin, Germany}
\affiliation{Technische Universität Berlin, Berlin, Germany}
\affiliation{RIKEN AIP, Tokyo, Japan}

\date{\today}

\begin{abstract}
Sequential minimal optimization methods, such as the Rotosolve and the Nakanishi--Fujii--Todo algorithm (NFT), are widely used for Variational Quantum Eigensolvers (VQEs). These methods optimize one parameter direction at a time, requiring measurements at only a few locations along that direction. In the presence of measurement shot noise, however, their performance depends critically on the choice of measurement locations, and recent studies suggest that equidistant measurements are optimal. However, we often observe that equidistant measurements are not always optimal in practice. We argue that this discrepancy between theory and practice arises from the fact that two assumptions underlying previous analyses do not generally hold: (1) the absence of prior knowledge about the energy minimizer, and (2) the use of the uncertainty of the estimated energy as a proxy for optimization performance. In this paper, we develop a new theory for determining optimal measurement locations. First, we show that incorporating prior information about the minimizer is beneficial. Early in optimization, when little is known about the pivot, i.e., the current minimizer, equidistant measurements are indeed near-optimal, but as the prior belief sharpens the optimal locations move away from equidistant. Second, rather than analyzing the uncertainty of the estimated minimum energy, we study the uncertainty of the estimator of the minimizer itself, which leads to substantially different strategies. Based on this analysis, we propose Prior-informed Adaptive Shifts (PAS), a method that automatically adjusts measurement locations during optimization. Numerical experiments across different shot counts and problems validate our theoretical findings and demonstrate that PAS adaptively recovers whichever fixed shift is best in each regime without it being specified in advance.
\end{abstract}

\maketitle

\section{\label{sec:introduction}Introduction}
Variational Quantum Eigensolvers (VQEs) are among the most promising applications of near-term quantum hardware, estimating the ground-state energy of a Hamiltonian by optimizing a parameterized quantum circuit \cite{peruzzo_variational_2014}. Applications include the simulation of quantum many-body dynamics \cite{smith_simulating_2019,fauseweh_quantum_2024}, lattice gauge theories \cite{meth_simulating_2025}, computation of molecular ground-states \cite{caditazi_folded_2024} or more general tasks such as combinatorial optimization \cite{moll_quantum_2018} and machine learning \cite{biamonte_quantum_2017}. The circuit runs on a quantum device and a classical optimizer adjusts its parameters, with the repeated energy measurements forming the main computational cost.

Each energy evaluation is estimated from a finite number of measurement shots and is therefore affected by statistical noise that shrinks as $\frac{1}{\sqrt{n}}$, where $n$ is the number of shots.
Since the total cost of a run is proportional to the number of shots taken, reducing the shots needed to reach a target accuracy translates directly into saving quantum resources. Unlike hardware imperfections, which active research aims to suppress \cite{cai_quantum_2023, acharya_quantum_2025, jiang_error_2024}, this shot noise is fundamental and persists even on fault-tolerant devices.

A widely used class of optimizers exploits the fact that, for common circuit ansätze, the energy, as a function of a single parameter, is a sinusoid fully determined by three coefficients \cite{nakanishi_sequential_2020, ostaszewski_structure_2021}. Sequential minimal optimization methods such as Rotosolve \cite{ostaszewski_structure_2021} and the Nakanishi--Fujii--Todo (NFT) algorithm \cite{nakanishi_sequential_2020} therefore optimize one parameter at a time, reconstructing this sinusoid from only three energy evaluations and moving to its minimum. In the noiseless limit the placement of these three evaluations is not relevant as long as they are distinct. Under shot noise, measurement locations become crucial as noise propagates through the reconstruction and its amplification depends on where the measurements are taken.

The shift angle $\alpha$ that controls this placement has thus far been treated as a fixed hyperparameter. \citet{endo_optimal_2023} analyzed how shot noise propagates to the estimated minimum \emph{energy} and recommended equidistant measurements spaced by $\frac{2\pi}{3}$, a choice that minimizes the condition number of the reconstruction and is supported by related analyses \cite{lai_interpolationbased_2026}. Other works adopt $\alpha = \frac{\pi}{2}$ without explicit justification \cite{nakanishi_sequential_2020,ostaszewski_structure_2021}. In practice, however, the equidistant choice is often not the best, and which fixed shift performs better depends on factors such as the problem and the shot budget.

We argue that this gap between theory and practice stems from two assumptions in prior analyses that do not generally hold. First, these analyses assume no prior knowledge of the minimizer. Yet sequential optimization accumulates exactly such knowledge as the optimization converges, since each parameter is revisited repeatedly and its estimates concentrate over time. Second, they take the uncertainty of the estimated energy as a proxy for optimization performance, whereas the quantity actually propagated between iterations is the estimated \emph{minimizer} of the energy, not the energy itself. The variance of the estimated minimizer depends on the shift through a different gradient than the energy variance, and leads to a different optimal shift.

Building on these two observations, we develop a new analysis of the optimal measurement locations and use it to design an adaptive method. Our contributions are as follows.

\begin{itemize}
    \item We place a von Mises prior over the location of the optimum and derive the expected variance of both the minimum energy and the minimizer estimators as functions of the shift $\alpha$ and the prior concentration. The equidistant recommendation of prior work is recovered as the special case of a uniform prior.

    \item We show that minimizing the variance of the minimizer, rather than the energy variance, yields a closed-form optimal shift that depends only on the prior concentration and is confined between $\frac{\pi}{2}$ and $\frac{2\pi}{3}$. The lower bound of $\frac{\pi}{2}$ is approached in the high-confidence limit and does not depend on the choice of prior.

    \item We propose Prior-informed Adaptive Shifts (PAS), which estimate the prior concentration from the history of past pivot estimates and sets the shift adaptively at each step, in global and gate-specific variants.

    \item On the TFIM and MaxCut problems, across high and low shot budgets, we show that PAS adaptively recovers whichever fixed shift performs best in each regime without it being specified in advance.
\end{itemize}

\section{\label{sec:background}Background}
\subsection{Variational Quantum Eigensolver}
\label{sec:vqe}

The Variational Quantum Eigensolver (VQE) \cite{peruzzo_variational_2014} is a hybrid quantum-classical protocol for estimating the ground-state energy of a given $Q$-qubit Hamiltonian $H$. A parameterized quantum circuit $U(\bftheta) = U_D \cdots U_1$, referred to as the ansatz, acts on an initial state $\rho_\mathrm{in}$ to produce the output state $\rho(\bftheta) = U(\bftheta) \rho_\mathrm{in} U^\dagger(\bftheta)$. While in general different choices for the specific structure are possible, we use a circuit consisting of either fixed entangling gates or parameterized gates of the form $U_d = \exp(i\frac{\theta_d}{2} O_d)$, where $\theta_d \in [-\pi,\pi)$ is a circuit parameter and $O_d$ is a Hermitian and unitary operator such as a tensor product of Pauli matrices. In this paper we only consider the case where each parameter enters at most one gate. The goal of VQE is to find

\begin{equation}
    \bftheta^* := \argmin_{\bftheta \in [-\pi,\pi)^D} \EH(\bftheta),
    \label{eq:vqe_min}
\end{equation}
where $\EH(\bftheta) = \Tr\left[H \rho(\bftheta) \right]$ is the expectation value of the Hamiltonian, which corresponds to the energy, and $D$ is the number of variational parameters. The optimal parameters $\bftheta^*$ map the initial state to the best approximation of the ground state achievable by the ansatz, i.e.,
\begin{equation}
    \rho(\bftheta^*) \approx \rho_\mathrm{GS}.
\end{equation}

In practice, the expectation value $\EH(\bftheta)$ cannot be evaluated exactly but must be estimated from a finite number of measurements. Each evaluation yields a noisy estimate
\begin{equation}
    f = \EH(\bftheta) + \epsilon_\mathrm{shots}, \quad \epsilon_\mathrm{shots} \sim \mathcal{N}\!\left(0,\, \frac{\sigma^2}{N_\mathrm{shots}}\right),
    \label{eq:noisy_obs}
\end{equation}
where $\sigma^2 = \mathrm{Var}_{\rho(\bftheta)}[H]$ is the variance of the Hamiltonian in the prepared state and $N_\mathrm{shots}$ is the number of measurement shots. 
We assume ideal, noiseless hardware throughout, where $\epsilon_\mathrm{shots}$ captures only the statistical noise from finite sampling, so $f$ is an unbiased estimate of $\EH(\bftheta)$.

While $\sigma$ in principle depends on $\bftheta$ \cite{mcclean_theory_2016,zhang_variational_2022}, we treat it as constant throughout this paper. Since the Pauli terms of $H$ generally do not commute, $H$ is measured by partitioning its Pauli terms into mutually commuting groups, each estimated in a separate set of circuit executions. Since these per-group estimates are statistically independent, the shot-noise variance of the energy estimate is governed by the sum of per-group variances $\sum_g \mathrm{Var}(H_g)$ rather than by $\mathrm{Var}(H)$. Thus, the inter-group covariances that make $\mathrm{Var}(H)$ vanish at eigenstates do not enter. Consequently the shot noise stays bounded away from zero even at eigenstates if $H$ decomposes into more than one commuting group. Rather than tracking this group- and $\bftheta$-dependence, we assume a constant $\sigma$, following \cite{nakanishi_sequential_2020, sung_using_2020, ostaszewski_structure_2021, tamiya_stochastic_2022, endo_optimal_2023, nicoli_physicsinformed_2023, anders_adaptive_2024}. This enables a closed-form analysis of the optimal measurement spacing. Incorporating the $\bftheta$-dependence of $\sigma$ into the measurement allocation is left for future work.

Since every energy evaluation consumes a fixed number of shots on a quantum device, the total cost of a VQE run is naturally expressed in terms of a shot budget. Reducing the number of shots required to reach a target accuracy in $\bftheta^*$ therefore directly translates to practical savings in quantum resources and is the main focus of this paper.

\subsection{Sequential Minimization for VQE}
\label{sec:seq_min}
For circuits of the form described above, it has been shown \cite{nakanishi_sequential_2020, ostaszewski_structure_2021} that the energy, viewed as a function of a single parameter $\theta_{d}$ while all other parameters $\bftheta_{-d}$ are kept fixed, reduces to a one-dimensional sinusoidal function
\begin{equation}
    \EH(\theta_d;\bftheta_{-d}) = b_{1,d} + \sqrt{2}\bigl(b_{2,d}\cos\theta_d + b_{3,d}\sin\theta_d\bigr),
    \label{eq:sinusoidal}
\end{equation}
which is uniquely determined by the three coefficients $\bfb_d = (b_{1,d},\, b_{2,d},\, b_{3,d})^T$ and minimized at
\begin{equation}
    \theta^*_d = \operatorname{atan2}(b_{3,d},\, b_{2,d}) + \pi,
    \label{eq:theta_opt}
\end{equation}
with a corresponding minimum energy of 
\begin{equation}
    f^*_d:= \EH(\theta_d^*;\,\bftheta_{-d}) = b_{1,d} - \sqrt{2}\sqrt{b_{2,d}^2 + b_{3,d}^2}.
     \label{eq:energy_opt}
\end{equation}

To determine $\bfb_d$, three energy evaluations at different parameter values suffice. We parameterize the measurement locations as $\bfPhi_d = (\varphi,\, \varphi+\alpha,\, \varphi-\alpha)$ where $\varphi \in \mbR$ is the pivot and $\alpha > 0$ is the angle between the measurements. 
This can be rewritten as the linear system
\begin{equation}
    \bff = A\bfb_d,
    \label{eq:linear_system}
\end{equation}
where $\bff = (f_1, f_2, f_3)^T$ are the three noisy energy observations and
\begin{equation}
    A = \begin{bmatrix}
        1 & \sqrt{2}\cos(\varphi) & \sqrt{2}\sin(\varphi) \\
        1 & \sqrt{2}\cos(\varphi + \alpha) & \sqrt{2}\sin(\varphi + \alpha) \\
        1 & \sqrt{2}\cos(\varphi - \alpha) & \sqrt{2}\sin(\varphi - \alpha)
    \end{bmatrix}.
    \label{eq:design_matrix}
\end{equation}

For $\alpha \in (0,\pi)$ the matrix $A$ is a nonsingular $3\times3$ matrix, so the coefficients are recovered by inversion, $\hat{\bfb}_d = A^{-1}\bff$.
Note that since the observations $\bff$ were obtained using a finite $N_\mathrm{shots}$ the estimator $\hat{\bfb}_d$ differs from the true coefficients $\bfb_d$.

Sequential minimization \cite{nakanishi_sequential_2020, ostaszewski_structure_2021} then proceeds by cycling through the $D$ parameters, at each step setting $\theta_d \leftarrow \hat{\theta}_d^*$, where $\hat{\theta}_d^*$ is obtained from \cref{eq:theta_opt} using the estimated coefficients $\hat{\bfb}_d$. Throughout this work, we consider the standard three-measurement procedure called Rotosolve by \cite{ostaszewski_structure_2021} depicted in \cref{alg:rotosolve}, where all observations in \cref{eq:linear_system} are obtained from fresh quantum measurements. This is in contrast to the work in \cite{nakanishi_sequential_2020}, which reuses the estimated minimum energy $\hat{f}^*_{d-1}$ from the previous parameter update as one of the three data points and thereby effectively carries out sequential minimal optimization (SMO) \cite{platt_sequential_1998}. While this effectively reduces the shot cost by $\frac{1}{3}$, it introduces correlated errors across sequential updates \cite{pedrielli_bias_2026}, which we do not address in this work.

\begin{algorithm}[H]
\small
\caption{Sequential Minimization for VQE}\label{alg:rotosolve}
\begin{algorithmic}
\Require Initial parameters $\bftheta \in [-\pi,\pi)^D$, shift $\alpha \in (0,\pi)$, number of measurements $N_\mathrm{shots}$
\Repeat
    \For{$d = 1, \dots, D$}
        \State Set pivot $\varphi \gets \theta_d$
        \State Measure $\bff \gets f(\bfPhi_d)$ using $N_\mathrm{shots}$ shots 
        \State Compute $\hat{\bfb}_d = A^{-1} \bff$ using \cref{eq:design_matrix}
        \State Compute $\hat{\theta}^*_d = \arctantwo(\hat{b}_{3,d},\, \hat{b}_{2,d})+\pi$ using \cref{eq:theta_opt}
        \State Update $\bftheta \gets (\hat{\theta}^*_d,\bftheta_{-d})$
    \EndFor
\Until{convergence or termination criterion is met, e.g.\ an exhausted shot budget}
\end{algorithmic}
\end{algorithm}

The offset $\alpha$ in \cref{eq:design_matrix} is a free parameter. Its choice affects the conditioning of the design matrix $A$ and, as we will show, the precision with which $\hat{\theta}^*_d$ can be estimated from noisy data. Determining the optimal $\alpha$ is the central question of this work.

\subsection{Prior Work on Measurement Shifts}
\label{sec:prior_work}

In the noiseless limit $N_\mathrm{shots} \to \infty$, the coefficients are recovered exactly $\bfb_d = \hat\bfb_d$ and the choice of shift $\alpha$ is irrelevant as long as the design matrix $A$ is invertible. With finite shots, however, the measurement noise $\bfepsilon_\mathrm{shots} \neq 0$ and propagates through the inverse to the estimated coefficients,
\begin{equation}
    \hat{\bfb}_d = A^{-1}\left(\EH(\bfPhi_d) + \bfepsilon_\mathrm{shots} \right) =  \bfb_d + A^{-1} \bfepsilon_\mathrm{shots},
    \label{eq:noisy_b}
\end{equation}
and the choice of $\alpha$ directly affects how this noise is amplified.

\citet{endo_optimal_2023} analyzed this propagation for the \emph{energy} estimate $\hat{f}^*_{d}$, showing that the variance of the estimated minimum energy depends on $\alpha$ through $A^{-1}$. Their analysis led to the recommendation of equally spaced shifts $\alpha = 2\pi/3$, which minimizes the condition number of $A$ and is optimal when the goal is to estimate $\hat{f}^*_{d}$ with minimal variance. Furthermore, \cite{lai_interpolationbased_2026} showed that, for equidistant frequencies, the equally spaced choice simultaneously minimizes the condition number, the mean squared error of the coefficient estimates, and the average variance of the reconstructed energy function.

However, in VQE the energy estimate at a single parameter update is not itself the quantity of interest. What matters for optimization is the \emph{minimizer} $\hat{\theta}^*_d$, since this determines the parameter update that is propagated to subsequent iterations. The variance of $\hat{\theta}^*_d$ involves a different gradient vector than the variance of $\hat{f}^*_d$, and consequently a different dependence on $\alpha$.

This distinction has not been addressed in prior work. \citet{anders_adaptive_2024} used $\alpha = 2\pi/3$ following the equal-spacing convention which minimizes the maximum uncertainty of the \emph{energy} estimate, while \cite{nakanishi_sequential_2020,ostaszewski_structure_2021} used $\alpha = \pi/2$ without explicit justification.

In the following section, we derive both the variance of $\hat{f}^*_d$ and  $\hat{\theta}^*_d$ as a function of $\alpha$ and show that the optimal shift depends on prior knowledge about the minimizer.

\section{\label{sec:analysis}Analysis of variances}
\subsection{Variance of minimum energy}
\label{sec:energy}

For notational brevity, we omit the subscript $d$ on $f$ and $\bfb$ throughout this section.
From \cref{eq:noisy_b}, the shot noise $\bfepsilon_\mathrm{shots}$ is mapped by the inverse $A^{-1}$ into the coefficient space. A first-order Taylor expansion of \cref{eq:energy_opt} around $\bfb$ then gives the estimated minimum energy as
\begin{equation}
    \hat f^*
    = 
    f^*
    +
    \bfv^T(A^{-1}\bfepsilon_\mathrm{shots})
    + 
    O_p(N_{\mathrm{shots}}^{-1}),
    \label{eq:energy_taylor}
\end{equation}
where the gradient
\begin{equation}
    \bfv = \left(1,\;\frac{-\sqrt{2} b_2}{r},\; \frac{-\sqrt{2} b_3}{r} \right)^T
    \label{eq:gradient_v}
\end{equation}
depends on the sinusoidal amplitude $r = \sqrt{b_2^2 + b_3^2}$. See \cref{app:gradient_v} for the full derivation.

The variance of the estimated minimum energy can then be expressed using \cref{eq:energy_taylor} (see \cref{app:var_general} for details) as
\begin{equation}
    \Var(\hat f^*)
    =
    \frac{\sigma^2}{N_\mathrm{shots}} \bfv^T P \bfv
    +
    O\!\left(N_\mathrm{shots}^{-2}\right),
    \label{eq:var_energy}
\end{equation}
where $P = (A^TA)^{-1}$. The variance depends on the shot noise as well as the geometry of spaced observations and the noise propagation gradient. Consequently, even under our assumption of constant $\sigma^2$, $\Var(\hat f^*)$ cannot be evaluated a priori, since it depends on the unknown coefficients $\bfb$ (via $\bfv$), which are only accessible after observing the energy function. To make the variance usable in practice we need a surrogate for $\bfb$. Iterative sequential optimization supplies one, since successive minimizer estimates concentrate around one another as optimization proceeds. This concentration is what we actually estimate, and it justifies a correspondingly sharper estimate of $\bfb$ and $\bfv$.

Since $v_1$ is constant and $\left(v_2, v_3\right)$ lies on a circle with radius $\sqrt{2}$ due to $v_2^2 + v_3^2 = 2$, we can write $v_2 = \sqrt{2}\cos(\theta^*)$ and $v_3 = \sqrt{2}\sin(\theta^*)$, using $\cos(\theta^*) = -b_2/r$ and $\sin(\theta^*) = -b_3/r$ from \cref{eq:theta_opt}. The gradient is therefore parameterized directly by the minimizer, $\bfv(\theta^*) = \left(1, \sqrt{2}\cos(\theta^*), \sqrt{2}\sin(\theta^*)\right)^T$, $\theta^* \in [-\pi,\pi)$.
We model our belief over $\theta^*$ with the von Mises distribution, which can be considered as the circular analogue to a Gaussian on the unit circle \cite{hillen_moments_2017},
\begin{equation}
    p_\mathrm{VM}(\theta^* | \mu, \kappa) = \frac{e^{\kappa \cos(\theta^* - \mu)}}{2\pi I_0(\kappa)}.
    \label{eq:vonmises_pdf}
\end{equation}
Here $\mu \in [-\pi,\pi)$ is its mean, $\kappa \ge 0$ the concentration which can be thought of loosely as the inverse of the variance for large $\kappa$ and $I_0(\kappa)$ is the modified Bessel function of zeroth-order. The von Mises prior reduces to a uniform prior with $\kappa=0$ and as $\kappa$ increases it becomes more concentrated around $\mu$. The concentration $\kappa$ then reflects our confidence in this estimate, growing as optimization progresses.

\begin{lemma}[Expected first-order variance]
\label{lemma:expvar}
Let $\hat q$ be the plug-in estimator of $q$ with gradient $\bfg(\theta^*)$ and let $\theta^* \sim p_\mathrm{VM}(\cdot \mid \mu, \kappa)$. Then
\begin{equation*}
    \E_{\theta^*}\bigl[\Var(\hat q)\bigr]
    =
    \frac{\sigma^2}{N_\mathrm{shots}} \Tr\bigl(P S_g\bigr)
    +
    O(N_\mathrm{shots}^{-2}),
\end{equation*}
where $S_g = \E_{\theta^*}\bigl[\bfg(\theta^*)\bfg(\theta^*)^T\bigr]$.
\end{lemma} The proof is given in \cref{app:expvar}.

Applying \cref{lemma:expvar} with $q = f^*$, so that $\bfg = \bfv$ of \cref{eq:gradient_v}, gives the following theorem:
\begin{theorem}[Expected variance of the minimum energy]
\label{thm:expvar_energy}
Let $\theta^* \sim p_\mathrm{VM}(\cdot \mid \mu, \kappa)$ with $\mu = \varphi$. For
observations at $\varphi, \varphi \pm \alpha$ with $\alpha \in (0,\pi)$, it holds that
\begin{equation}
    \E_{\theta^*}\bigl[\Var(\hat f^*)\bigr]
    =
    \frac{\sigma^2}{N_\mathrm{shots}} \Tr\bigl(P S_v\bigr)
    + O(N_\mathrm{shots}^{-2}).
    \label{eq:expvar_energy}
\end{equation}
\end{theorem}
The proof is given in \cref{app:expvar_energy}. \Cref{eq:expvar_energy} depends on the data only through $\sigma^2/N_\mathrm{shots}$, the
spacing $\alpha$, and the concentration $\kappa$.

\Cref{fig:validation_energy} validates \cref{eq:expvar_energy} against the empirical expected variance obtained from 50\,000 sampled minimizers per $(\kappa,\alpha)$ pair, showing close agreement across all tested values of $\alpha$ and $\kappa$.

\begin{figure}[ht]
    \centering
    \includegraphics[width=0.9\linewidth]{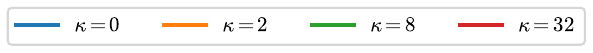}
    \includegraphics[width=0.9\linewidth]{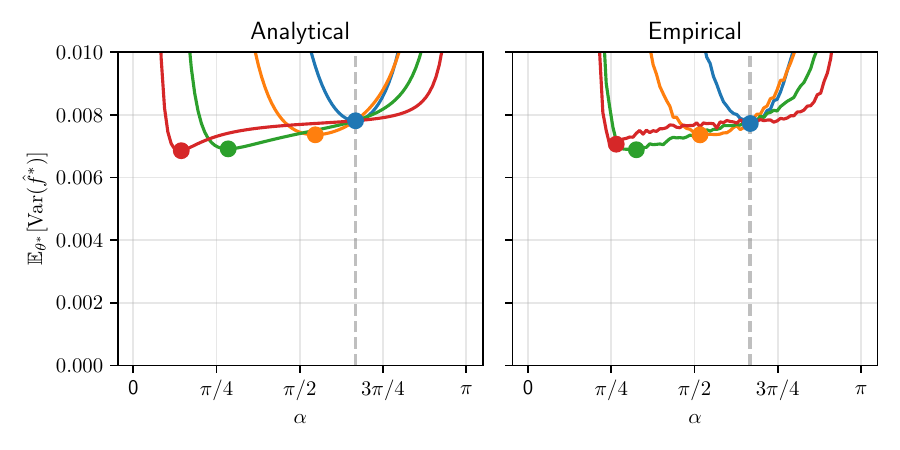}
    \caption{Theoretical (left) and empirical (right) expected variance of $\hat{f}^*$ for a single parameter update, as a function of $\alpha$ for several values of $\kappa$. The dashed line indicates $\alpha = 2\pi/3$. For each $(\kappa, \alpha)$ pair, 50\,000 minimizer $\theta^*_i \sim \mathrm{VonMises}(\mu{=}0,\, \kappa)$ are drawn and noisy observations at $(0, -\alpha, +\alpha)$ are simulated with $\mathcal{N}(0,\, \sigma^2 / N_\mathrm{shots})$ noise. Settings: $\sigma = 1$, $N_\mathrm{shots} = 128$.}
    \label{fig:validation_energy}
\end{figure}

\citet{endo_optimal_2023} showed using a uniform prior, i.e., $\kappa=0$, that equidistant measurements spaced by $2\pi/3$ minimize \cref{eq:expvar_energy}. However, for arbitrary $\kappa$, the optimal positioning of observations cannot be obtained in closed form. This limitation is less severe than it appears, because in sequential optimization for VQE we usually care more about minimizing the error of $\theta^*$ than of $f^*$ as this is the quantity we compute on each iteration in \cref{alg:rotosolve} and ultimately should yield the solution of \cref{eq:vqe_min}. In the next section we show how we can perform the same steps to minimize the expected variance of the minimizer.

\subsection{Variance of minimizer}
\label{sec:minimizer}
Similar to \cref{sec:energy} we now analyze the variance of the estimated minimizer $\hat \theta^*$. Substituting \cref{eq:noisy_b} into \cref{eq:theta_opt} gives
\begin{equation}
    \hat \theta^* 
    =
    \theta^* + \bfw^T \left(A^{-1} \bfepsilon_\mathrm{shots} \right)
    +
    O_p(N_{\mathrm{shots}}^{-1}),
    \label{eq:theta_taylor}
\end{equation}
where the gradient
\begin{equation}
    \bfw = \left(0,\frac{-b_3}{r^2},\frac{b_2}{r^2} \right)^T
    \label{eq:gradient_w}
\end{equation}
depends on the sinusoidal amplitude $r = \sqrt{b_2^2 + b_3^2}$. Compared to $\bfv$ in \cref{eq:gradient_v}, the first entry of $\bfw$ is $0$ because $\theta^*$ does not depend on $b_1$, and the roles of $b_2$ and $b_3$ are interchanged with one sign flip. The $(w_2, w_3)$ components are therefore tangential to a circle of radius $r$ rather than radial as in $\bfv$, which reflects that $f^*$ depends on the amplitude $r$ while $\theta^*$ depends on the phase. See \cref{app:gradient_w} for the full derivation.

Analogous to \cref{eq:var_energy} we express the variance of the estimated minimizer using \cref{eq:theta_taylor}. Then the same derivation as in \cref{app:var_general} provides
\begin{equation}
    \Var(\hat \theta^*)
    =
    \frac{\sigma^2}{N_\mathrm{shots}} \bfw^T P \bfw
    + 
    O\!\left(N_\mathrm{shots}^{-2}\right).
    \label{eq:var_theta}
\end{equation}
Now we can place again a von Mises prior to model our prior belief where the minimizer $\theta^*$ might be. The reparameterized gradient is $\bfw(\theta^*) = \frac{1}{r}\left(0, \sin(\theta^*), -\cos(\theta^*)\right)^T$.

Applying \cref{lemma:expvar} with $q = \theta^*$, so that $\bfg = \bfw$ of
\cref{eq:gradient_w}, gives the following theorem:
\begin{theorem}[Expected variance of the minimizer]
\label{thm:expvar_theta}
Let $\theta^* \sim p_\mathrm{VM}(\cdot \mid \mu, \kappa)$ with $\mu = \varphi$ and observations at $\varphi, \varphi \pm \alpha$ with $\alpha \in (0,\pi)$. Then, it holds that
\begin{equation}
    \E_{\theta^*}\bigl[\Var(\hat\theta^*)\bigr]
    =
    \frac{\sigma^2}{N_\mathrm{shots}} \Tr\bigl(P S_w\bigr)
    + O(N_\mathrm{shots}^{-2}).
    \label{eq:expvar_theta}
\end{equation}
\end{theorem}
\Cref{eq:expvar_theta} depends only on $\sigma^2/N_\mathrm{shots}$, the spacing $\alpha$, the amplitude $r$ and the concentration $\kappa$.
The proof is given in \cref{app:expvar_theta}.

In contrast to \cref{eq:expvar_energy}, the minimizer of \cref{eq:expvar_theta} can be obtained in a closed form.

\begin{theorem}[Optimal shift angle]
\label{thm:optimal_shift}
For $\theta^* \sim p_\mathrm{VM}(\cdot \mid \mu, \kappa)$, the shift minimizing the expected first-order variance of the minimizer \cref{eq:expvar_theta} is
\begin{equation}
    \alpha^*
    =
    \argmin_\alpha \E_{\theta^*}\bigl[\Var(\hat\theta^*)\bigr]
    = 2\arctan\sqrt{\frac{1+\eta}{2}},
    \label{eq:optimal_shift}
\end{equation}
where $\eta = \sqrt{\frac{25-23 R_2(\kappa)}{1+R_2(\kappa)}}$ and $R_2(\kappa) := I_2(\kappa)/I_0(\kappa)$ denotes the ratio of modified Bessel functions of the first kind. 
\end{theorem}
The proof is given in \cref{app:alpha_closedform}.
Thus, we can efficiently adapt the spacing between observations to minimize the error in $\hat \theta^*$. Note that \cref{eq:optimal_shift} depends only on $\kappa$, which quantifies how tightly our belief about $\theta^*$ concentrates around the current pivot $\varphi$.

By comparing both \Cref{fig:validation_energy,fig:validation_theta} it can be seen that for $\kappa=0$ both objectives yield equidistant shifts $\alpha^* = 2\pi/3$ between measurements. On the other hand, when $\kappa \to \infty$, $\alpha^* \to 0$ if $\Var(\hat{f}^*)$ is minimized, while $\alpha \to \frac{\pi}{2}$ if $\Var(\hat{\theta}^*)$ is minimized. The optimal shifts are actually confined:

\begin{corollary}[Bound on the optimal shift angle]
\label{cor:bound}
For every $\kappa \ge 0$, the optimal shift angle \cref{eq:optimal_shift} is bounded as
\begin{equation*}
    \frac{\pi}{2} \;\le\; \alpha^* \;\le\; \frac{2\pi}{3}.
\end{equation*}
\end{corollary}
The proof is given in \cref{app:bound}.

\begin{figure}[ht]
    \centering
    \includegraphics[width=0.9\linewidth]{figures/validation/legend_kappas.pdf}
    \includegraphics[width=0.9\linewidth]{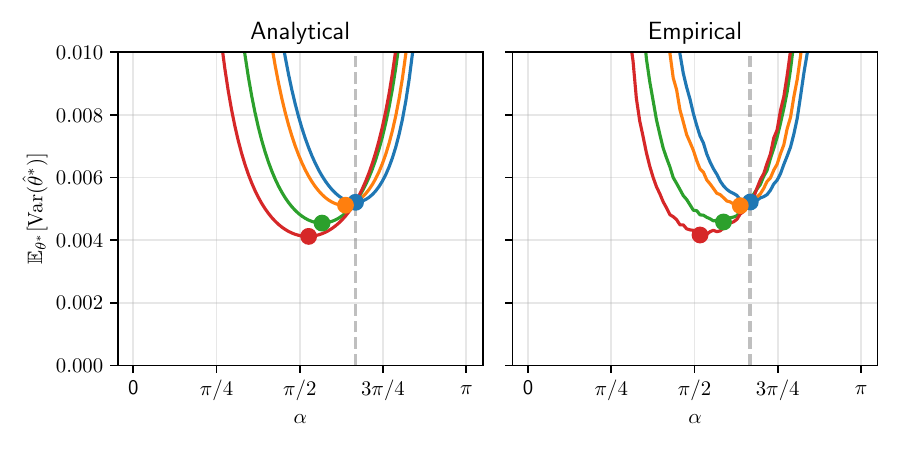}
    \caption{Theoretical (left) and empirical (right) expected variance of $\hat\theta^*$ as a function of $\alpha$ for several values of $\kappa$. The dashed line indicates $\alpha = 2\pi/3$. For each $(\kappa, \alpha)$ pair, 50\,000 minimizer $\theta^*_i \sim \mathrm{VonMises}(\mu{=}0,\, \kappa)$ are drawn and noisy observations at $(0, -\alpha, +\alpha)$ are simulated with $\mathcal{N}(0,\, \sigma^2 / N_\mathrm{shots})$ noise. Angular errors are wrapped to $[-\pi, \pi]$. Settings: $r = 1$, $\sigma = 1$, $N_\mathrm{shots} = 128$.}
    \label{fig:validation_theta}
\end{figure}

To understand why $\frac{\pi}{2}$ acts as a lower bound, we examine the high-confidence limit $\kappa \to \infty$ in which the pivot $\varphi$ converges to $\theta^*$. In this regime $\bfw$ in \cref{eq:gradient_w} aligns with the $b_3$-axis, so only the sine component of the model informs $\hat\theta^*$, estimated through the antisymmetric difference of the observations at $\varphi, \varphi \pm \alpha$ only. As shown in \cref{app:pivot_aligned}, this yields the following asymptotic form:

\begin{theorem}[Asymptotic minimizer variance]
\label{thm:asymptotic}
In the limit $\varphi \to \theta^*$, the first-order variance of the estimated minimizer reduces to
\begin{equation*}
    \Var(\hat\theta^*) = \frac{\sigma^2}{N_\mathrm{shots}\, 4 r^2 \sin^2(\alpha)} + O(N_\mathrm{shots}^{-2}),
\end{equation*}
which is uniquely minimized over $\alpha \in (0, \pi)$ at $\alpha = \frac{\pi}{2}$.
\end{theorem}
The resulting optimal shift $\alpha = \frac{\pi}{2}$ coincides with the optimal shift \cite{pedrielli_bayesian_2025} of the parameter-shift rule, which estimates the gradient by extracting the antisymmetric component.

The $\frac{\pi}{2}$ bound is not tied to the von Mises prior. Once the pivot sits at $\theta^*$, the two symmetric observations agree, and only a residual offset makes them differ. Their difference responds most strongly to such an offset where the sinusoid is steepest, a quarter period away from the pivot, while the noise on it stays independent of $\alpha$.
By contrast, $\hat f^*$ depends on $b_1$ and the amplitude $r$ only through $b_1 - \sqrt{2}\,r$, which all three observations inform once they concentrate around $\theta^*$, where the fit is evaluated. No contrast between them is needed, which explains why $\Var(\hat f^*)$ continues to favor $\alpha \to 0$ as $\kappa$ grows.

\section{\label{sec:method}Prior-informed Adaptive Shifts (PAS)}
We now introduce \textit{Prior-informed Adaptive Shifts} (PAS), which uses the history of past pivot estimates to set the measurement spacing $\alpha$ adaptively via \cref{eq:optimal_shift}. The only required quantity is the concentration $\kappa$ of the von Mises prior, which encodes our confidence that the current pivot $\varphi$ is close to $\theta^*_d$.

We estimate $\kappa$ from the history of previous pivot estimates. Let $\mathcal{H}_d =  \{\hat\theta^*_{d,\tau}\}_{\tau=t-W}^{t-1}$ denote the estimates obtained for parameter $d$ 
over the preceding $W$ iterations. The maximum-likelihood estimator $\hat\kappa$ under a von 
Mises model satisfies
\begin{equation}
    \frac{I_1(\hat\kappa)}{I_0(\hat\kappa)} = \bar{R}, \qquad
    \bar{R} = \frac{1}{W}\left|\sum_{\tau=t-W}^{t-1} e^{i\hat\theta^*_{d,\tau}}\right|,
    \label{eq:kappa_mle}
\end{equation}
where $\bar{R} \in [0,1]$ is the mean resultant length of the circular sample and $I_0, I_1$ are modified Bessel functions of zeroth and first order. \Cref{eq:kappa_mle} has no closed form but $I_1(\kappa)/I_0(\kappa)$ is monotonically increasing in $\kappa$, so $\hat\kappa$ is obtained by standard nonlinear root-finding. As optimization converges and prior belief concentrates, $\bar R \to 1$ and $\hat\kappa \to \infty$, so $R_2(\hat\kappa) \to 1$ and \cref{eq:optimal_shift} drives the measurement spacing toward $\pi/2$.

Two natural ways exist for how $\kappa$ is estimated.
\begin{itemize}
    \item A gate-specific estimate uses only the history of parameter $d$ itself, capturing that some parameters are pinned down early while others remain dispersed throughout the optimization.
    \item A \textit{global} estimate pools histories across all $D$ parameters, which reduces variance at the cost of individual resolution and is particularly useful in early iterations when little per-parameter history has accumulated.
\end{itemize}
The full procedure is given in \cref{alg:pas}.

\begin{figure*}
\begin{minipage}{\linewidth}
\begin{algorithm}[H]
\small
\caption{PAS-Rotosolve}\label{alg:pas}
\begin{algorithmic}
\Require Initial parameters $\bftheta \in [-\pi,\pi)^D$, number of measurements $N_\mathrm{shots}$, window size $W$, estimation mode $m \in \{\textsc{global},\, \textsc{gate}\}$
\State Initialize buffers $\mathcal{H}_d \gets [\hspace{0.5em}]$ of capacity $W$ for all $d$
\Repeat
    \For{$d = 1, \dots, D$}
        \If{$m = \textsc{global}$ \textbf{or} $|\mathcal{H}_d| < 2$}
            \State Estimate $\hat\kappa$ from $\bigcup_{d'=1}^{D} \mathcal{H}_{d'}$, or $\hat\kappa \gets 0$ if all buffers empty \Comment{Global}
        \Else
            \State Estimate $\hat\kappa$ from $\mathcal{H}_d$ \Comment{Gate-specific}
        \EndIf
        \State $\hat\alpha \gets \argmin_\alpha \E_{\theta^*}\bigl[\Var(\hat\theta_d^*)\bigr]$ using $\hat\kappa$
        \State Set pivot $\varphi \gets \theta_d$
        \State Measure $\bff \gets f(\bfPhi_d)$ with $N_\mathrm{shots}$ shots, shift $\hat\alpha$
        \State $\hat{\bfb}_d \gets A^{-1}(\varphi, \hat\alpha)\, \bff$
        \State $\hat\theta^*_d \gets \arctantwo(\hat{b}_{3,d},\, \hat{b}_{2,d}) + \pi$
        \State Update $\theta_d \gets \hat\theta^*_d$ and append $\hat\theta^*_d$ to $\mathcal{H}_d$
    \EndFor
\Until{convergence or termination criterion is met, e.g.\ an exhausted shot budget}
\end{algorithmic}
\textit{$\hat\kappa$-estimation via \cref{eq:kappa_mle}; $\hat\alpha$ via \cref{eq:optimal_shift}; $\hat{\bfb}_d$ via \cref{eq:noisy_b}; $\hat\theta_d^*$ via \cref{eq:theta_opt}.}

\end{algorithm}
\end{minipage}
\end{figure*}

\section{\label{sec:experiments}Experiments}
\subsection{Experimental Setup}
\label{sec:experimental_setup}

\begin{figure*}[t]
    \centering
    \includegraphics[width=0.9\linewidth]{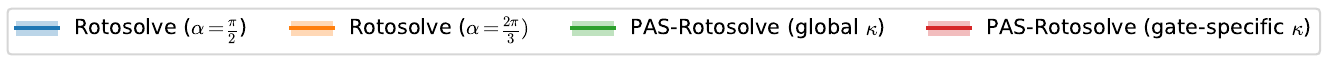}\\[3pt]
    \begin{subfigure}{\linewidth}
        \centering
        \includegraphics[width=\linewidth]{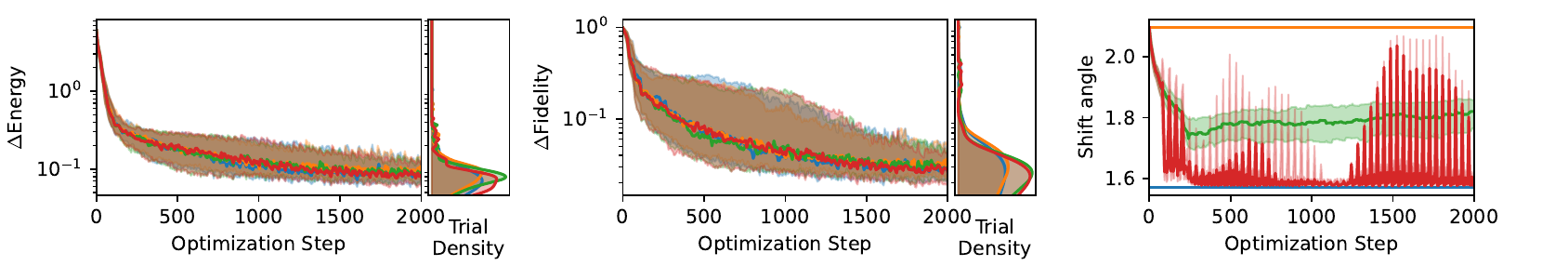}
        \caption{$N_\mathrm{shots}=1000$, 2000 steps}
        \label{fig:tfim_high}
    \end{subfigure}\\[3pt]
    \begin{subfigure}{\linewidth}
        \centering
        \includegraphics[width=\linewidth]{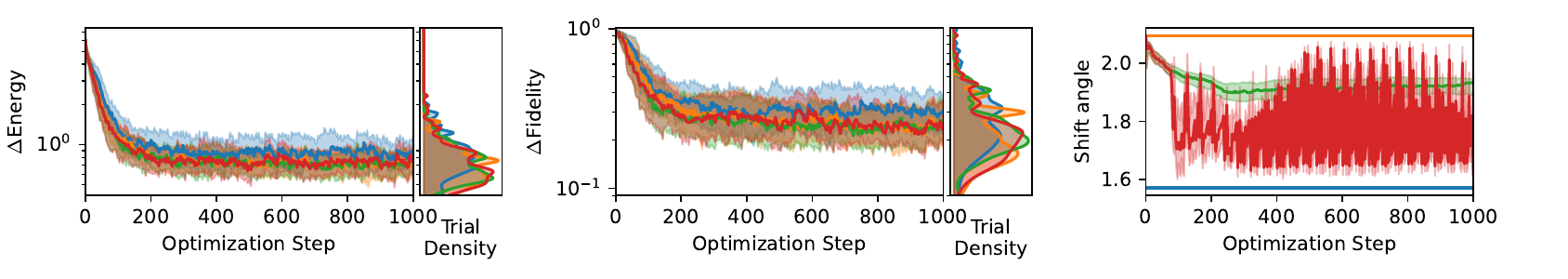}
        \caption{$N_\mathrm{shots}=100$, 1000 steps}
        \label{fig:tfim_low}
    \end{subfigure}
    \caption{VQE on the TFIM Hamiltonian, $L=3$-layer $Q=5$-qubit ESU2 circuit.
      Rows: shot budget; columns within a row: energy (left), fidelity (middle), shift angle $\alpha$ (right).
      \methoddesc{Rotosolve} \trialdesc}
    \label{fig:tfim}
\end{figure*}

\begin{figure*}[t]
    \centering
    \includegraphics[width=0.9\linewidth]{figures/experiments/rotosolve_legend.pdf}\\[3pt]
    \begin{subfigure}{\linewidth}\centering
        \includegraphics[width=\linewidth]{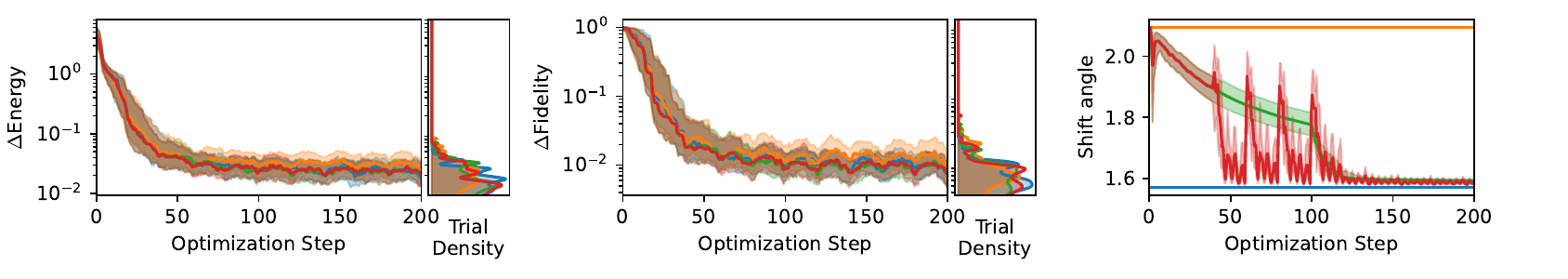}
        \caption{$N_\mathrm{shots}=200$, 200 steps}
        \label{fig:maxcut_high}
    \end{subfigure}\\[3pt]
    \begin{subfigure}{\linewidth}\centering
        \includegraphics[width=\linewidth]{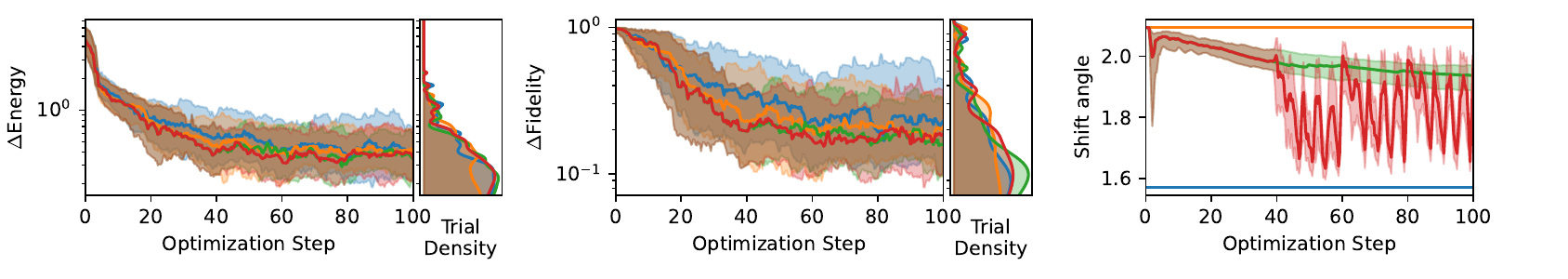}
        \caption{$N_\mathrm{shots}=20$, 100 steps}
        \label{fig:maxcut_low}
    \end{subfigure}
    \caption{VQE on the MaxCut Hamiltonian, $L=5$-layer $Q=4$-qubit HEA circuit.
      Same layout as \cref{fig:tfim}. \methoddesc{Rotosolve} \trialdesc}
    \label{fig:maxcut}
\end{figure*}

We evaluate PAS-Rotosolve against fixed-shift Rotosolve on two problem instances that test complementary aspects of our method. For all experiments, we prepare $100$ different random initial points sampled from the uniform distribution $\mathcal{U}[-\pi,\pi)^D$ that are used as a starting point for all optimization methods. Our Python implementation uses Qiskit~\cite{javadi-abhari_quantum_2024} for the classical simulation of quantum hardware. The implementation for reproducing our results is attached as supplemental material.

\paragraph{Hamiltonians and Quantum Circuits.}
We consider two target Hamiltonians. The first is the transverse-field Ising model (TFIM) at the ferromagnetic critical point,
\begin{equation}
    H_\mathrm{TFIM} = -J\sum_{j=1}^{Q-1} Z_j Z_{j+1} - h \sum_{j=1}^{Q} X_j,
    \label{eq:tfim}
\end{equation}
with $J = 1$, $h = 1$ and open boundary conditions on $Q = 5$ qubits. The TFIM is a widely used benchmark for assessing VQE performance due to its practical relevance and the availability of an analytically computable ground state for small system sizes. We pair this Hamiltonian with an Efficient SU(2) ansatz containing $L=3$ layers, so the total number of variational parameters is $D = 2Q(L + 1) = 40$ (see \cref{fig:circuit_esu2} for further details).
 
The second is a MaxCut Hamiltonian on an undirected graph $G = (V, E)$,
\begin{equation}
    H_{\mathrm{MaxCut}} = \sum_{i=1}^{Q}\sum_{j=i+1}^{Q} \frac{w_{ij}}{2}\left(Z_i Z_j - I\right) + \lambda\left(I - Z_1\right),
    \label{eq:maxcut_hamiltonian}
\end{equation}
where $w_{ij} = 1$ if $(i,j) \in E$ and $0$ otherwise, and any $\lambda>0$ breaks the tie
between a partition and its complement, without changing the optimal cut. We consider the graph with $V = \{1,2,3,4\}$ and $E = \{(1,2),(1,3),(1,4),(2,3),(3,4)\}$ with $\lambda=3$, following the setup of~\cite{lai_interpolationbased_2026}. This amounts to a problem of $Q = 4$ qubits. We pair this Hamiltonian with a $L = 5$ layer hardware-efficient ansatz (HEA). The total number of variational parameters is $D = QL = 20$ (see \cref{fig:circuit_hea} for further details).

\paragraph{Evaluation Metrics.}
We compare all methods using the energy gap
\begin{equation}
    \Delta\text{Energy} = \EH(\hat{\bftheta}) - \Tr \left[H \rho_\mathrm{GS} \right]\,,
    \label{eq:delta_energy}
\end{equation}
and the infidelity
\begin{equation}
    \Delta\text{Fidelity} = 1 - \Tr\left[\rho_\mathrm{GS}\, \rho(\hat{\bftheta})\right]\,,
    \label{eq:delta_fidelity}
\end{equation}
both plotted on a log scale (smaller is better) as the main metrics. Here, $\rho_\mathrm{GS}$ and $\Tr \left[H \rho_\mathrm{GS} \right]$ are the ground-state density operator and ground-state energy, respectively, both computed by exact diagonalization of the target Hamiltonian~$H$. As a measure of quantum computational cost, we report the number of optimization steps.

\paragraph{Baseline Methods.}
We compare PAS-Rotosolve to standard Rotosolve~\cite{ostaszewski_structure_2021} with fixed shifts $\alpha=\frac{2\pi}{3}$ and $\alpha=\frac{\pi}{2}$. As mentioned in~\cref{sec:prior_work}, both settings are common in the literature and are also the boundaries of $\alpha$ when using PAS-Rotosolve. As discussed in~\cref{sec:seq_min}, each Rotosolve step requires observations at three distinct positions along the current parameter direction. Unlike NFT~\cite{nakanishi_sequential_2020}, which reuses a previously measured energy value at the current optimum, we perform fresh measurements at all three positions in every step for all methods. The only difference in all methods is therefore the difference in the shifts $\alpha$.

\paragraph{Algorithm Settings.}
We measure each observation with a fixed $N_\mathrm{shots}$, hence for a single step in the optimization procedure we measure $3\times N_\mathrm{shots}$. We use the window size $W = 5$ to estimate $\kappa$. In \textit{gate-specific} mode, estimating $\hat\kappa$ from $\mathcal{H}_d$ requires at least two pivot estimates, so the first sweep over all parameters falls back to a \textit{global} estimate before switching. Rather than fixing $\hat\kappa = 0$ during this warm-up, we use the global estimate, which pools pivots across all $D$ parameters and is therefore available from the first sweep onward, providing a data-driven $\hat\kappa$ 
even before any per-parameter buffer is populated.

\subsection{Results}
\label{sec:results}

\paragraph{Transverse-field Ising model.}
When measuring with a relatively large shot count per iteration (\cref{fig:tfim_high}) the two fixed-shift Rotosolve variants ($\alpha=\frac{\pi}{2}$ and $\alpha=\frac{2\pi}{3}$) perform almost identically in both energy and fidelity, so the choice of shift carries little signal for an adaptive scheme to exploit. PAS matches the best fixed shift without tuning. With gate-specific $\kappa$, most gates are resolved with high certainty (large $\kappa$) and settle near $\alpha=\frac{\pi}{2}$, while a minority with poorly estimated pivots oscillate, yielding small $\kappa$ and shifts pulled toward the equidistant $\frac{2\pi}{3}$. The single global-$\kappa$ variant averages over this heterogeneity and settles at an intermediate $\alpha\approx1.8$.

With a smaller shot count per iteration (\cref{fig:tfim_low}) the methods separate. Fixed-shift Rotosolve at $\alpha=\frac{\pi}{2}$ is clearly worse in both energy and fidelity, whereas equidistant Rotosolve ($\alpha=\frac{2\pi}{3}$) and both PAS variants do better. This matches our theory: the more noisy setup inflates the uncertainty on the pivot, lowering $\kappa$ and favouring more equidistant shifts. PAS recovers this regime adaptively, raising $\alpha$ toward $\frac{2\pi}{3}$ (PAS (global $\kappa$) rises to $\alpha\approx1.9$) relative to \cref{fig:tfim_high}.

\paragraph{MaxCut.}
The same pattern holds on MaxCut. At the larger budget (\cref{fig:maxcut_high}) Rotosolve $(\alpha=\frac{\pi}{2})$ performs best among the fixed shifts, with Rotosolve $(\alpha=\frac{2\pi}{3})$ trailing slightly, and both PAS variants match the best fixed shift without it being specified a priori. The shift angle panel makes the mechanism visible: global $\kappa$ drives $\alpha$ smoothly down toward $\frac{\pi}{2}$, and gate-specific $\kappa$ does the same after an initial oscillatory phase, collapsing onto $\alpha\approx\frac{\pi}{2}$ once the pivots are well resolved (large $\kappa$).

At the smaller budget (\cref{fig:maxcut_low}) the ordering reverses. Equidistant Rotosolve $(\alpha=\frac{2\pi}{3})$ now performs better among the fixed shifts, as expected from the same pivot-uncertainty argument: the larger shot noise lowers $\kappa$ and favours more equidistant shifts. The PAS variants track this adaptively, with global $\kappa$ settling near $\alpha\approx1.9$. The gate-specific $\kappa$ shifted upward relative to \cref{fig:maxcut_high} but does not converge to the same shift angle for all parameters, most likely due to the larger energy variance. In fidelity their medians settle slightly below both fixed shifts, suggesting PAS can improve on either fixed choice rather than only matching it.

Two ablations confirm the bound at the full-optimization level. With linearly interpolated schedules (\cref{app:shift_linear}), runs terminating at $\alpha=\frac{\pi}{2}$ remain stable while smaller targets degrade or collapse. With fixed shifts (\cref{app:shift_fixed}), every $\alpha<\frac{\pi}{2}$ drastically degrades performance, increasingly so for smaller $\alpha$.

\section{\label{sec:conclusion}Conclusion}
The choice of shift angle in sequential minimal optimization methods such as Rotosolve and the Nakanishi--Fujii--Todo (NFT) algorithms has so far been treated as a fixed design decision, with equidistant measurements as the standard recommendation \cite{anders_adaptive_2024,endo_optimal_2023,lai_interpolationbased_2026}. Our analysis reveals this recommendation to be a special case. Placing a von Mises prior on the minimizer, equidistant shifts are optimal precisely when the prior is uniform, i.e., at the beginning of the optimization procedure. Once knowledge about the pivot accumulates, the optimal spacing tightens around it.

Central to this result is a distinction between two objectives. What matters in sequential optimization is the variance of the estimated \emph{minimizer}, not of the estimated optimal \emph{energy}, since the minimizer is the quantity updated at every iteration which influences the energy measurements. As the confidence grows the two diverge. Optimizing the minimum energy variance drives the spacing to zero, whereas for the variance of the minimizer this is bounded below by $\frac{\pi}{2}$. Consequently, the optimal shift is always confined to $\left[\frac{\pi}{2}, \frac{2\pi}{3}\right]$ and admits a closed form depending only on the prior concentration. Our first-order variance approximations matched empirical values across concentrations. We proposed Prior-informed Adaptive Shifts (PAS) as a practical way of using this observation. By inferring the concentration from the history of past pivot estimates, either pooled globally or per parameter, PAS sets the shift at each step without manual tuning. On TFIM and MaxCut, across high and low shot budgets, PAS recovered whichever fixed shift ($\frac{2\pi}{3}$ or $\frac{\pi}{2}$) performed best. The optimal shift angle is therefore not a property of the algorithm but of the state and does not need to be chosen by hand.

\begin{acknowledgments}
This work was funded by the German Ministry for Education and Research (BMBF) under the grant BIFOLD25B
and by the European Union’s HORIZON MSCA Doctoral Networks program project AQTIVATE (101072344).
This work was also supported by the European Union’s Horizon Europe Framework Programme (HORIZON) under the ERA Chair scheme with grant agreement no. 101087126,
and by the Ministry of Science, Research and Culture of the State of Brandenburg within the Center for Quantum Technology and Applications (CQTA). 
\begin{center}
    \includegraphics[width = 0.3\textwidth]{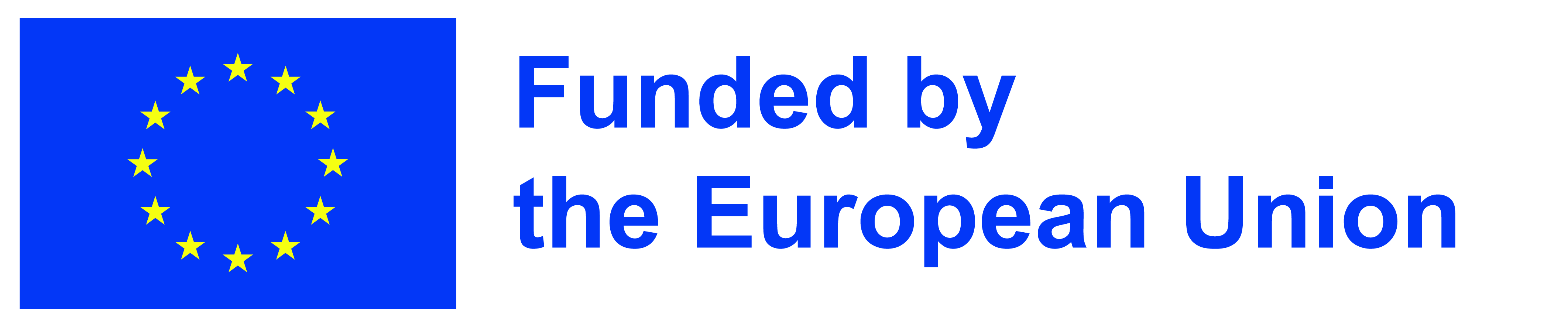}
    \includegraphics[width = 0.08\textwidth]{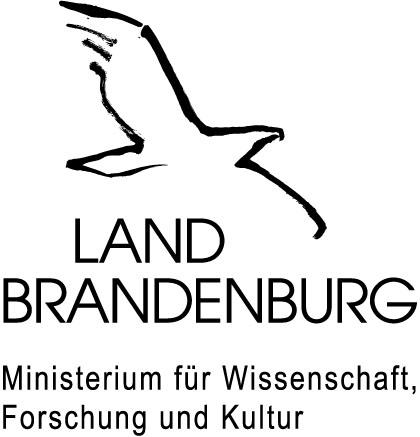}
\end{center}
\end{acknowledgments}

\appendix
\crefalias{section}{appendix}
\crefalias{subsection}{appendix}

\section{Noise propagation gradient for optimal estimates}
\subsection{\label{app:gradient_v} Estimated minimum energy}
The minimum energy along dimension $d$ is
\begin{equation*}
    f^*(\bfb') = b_1' - \sqrt{2}\sqrt{{b_2'}^{2} + {b_3'}^{2}}.
\end{equation*}
To determine how coefficient noise propagates to this estimate, we evaluate the gradient at the true coefficients $\bfb$:
\begin{align*}
    \bfv &= \nabla f^*(\bfb')\big|_{\bfb'=\bfb} \\
    &= \left. \left(1,\; \frac{-\sqrt{2}{b_2'}}{\sqrt{{b_2'}^2 + {b_3'}^2}},\; \frac{-\sqrt{2}{b_3'}}{\sqrt{{b_2'}^2 + {b_3'}^2}} \right)^T\right|_{\bfb'=\bfb} \\
    &= \left(1,\; \frac{-\sqrt{2}b_2}{r},\; \frac{-\sqrt{2}b_3}{r} \right)^T,
\end{align*}
where $r = \sqrt{b_2^2 + b_3^2}$ is the sinusoidal amplitude.
A first-order Taylor expansion of \cref{eq:energy_opt} combined with \cref{eq:noisy_b} then gives
\begin{align*}
    \hat f^* &\approx f^* + \bfv^T \left(\hat\bfb - \bfb \right) \\
    &= f^* + \bfv^T \left(\bfb + A^{-1}\bfepsilon_\mathrm{shots} - \bfb \right) \\
    &= f^* + \bfv^T A^{-1}\bfepsilon_\mathrm{shots},
\end{align*}
which is \cref{eq:energy_taylor}.

\subsection{\label{app:gradient_w} Estimated minimizer}
The same approach applies to the energy minimizer
\begin{equation*}
    \theta^*(\bfb') = \arctantwo(b_3',b_2') + \pi.
\end{equation*}
Here the gradient takes a different form because $\theta^*$ is insensitive to $b_1$ and depends on the ratio of $b_3$ and $b_2$ rather than their magnitude:
\begin{align*}
    \bfw &= \nabla \theta^*(\bfb')\big|_{\bfb'=\bfb} \\
    &= \left. \left(0,\; \frac{-b_3'}{{b_2'}^2 + {b_3'}^2},\; \frac{b_2'}{{b_2'}^2 + {b_3'}^2} \right)^T\right|_{\bfb'=\bfb} \\
    &= \left(0,\; \frac{-b_3}{r^2},\; \frac{b_2}{r^2} \right)^T.
\end{align*}
Expanding \cref{eq:theta_opt} to first order as before yields \cref{eq:theta_taylor}:
\begin{align*}
    \hat \theta^* &\approx \theta^* + \bfw^T \left(\hat\bfb - \bfb \right) \\
    &= \theta^* + \bfw^T \left(\bfb + A^{-1}\bfepsilon_\mathrm{shots} - \bfb \right) \\
    &= \theta^* + \bfw^T A^{-1}\bfepsilon_\mathrm{shots}.
\end{align*}

\section{Variance of estimated quantity}
\label{app:var_general}
Let $\hat q \approx q + \bfg(\theta^*)^T(A^{-1} \bfepsilon_\mathrm{shots})$
be a first-order Taylor approximation of some estimated quantity $\hat q$
around its true value $q$, where $\bfg(\theta^*)$ is the corresponding
gradient. Then
\begin{align}
    \mathrm{Var}(\hat{q}) &= \E\bigl[(\hat{q} - q)^2\bigr] \nonumber \\
    &\approx \E\bigl[(\bfg(\theta^*)^T A^{-1} \bfepsilon_\mathrm{shots})^2\bigr] \nonumber \\
    &= \E\bigl[\bfg(\theta^*)^T A^{-1} \bfepsilon_\mathrm{shots} \bfepsilon_\mathrm{shots}^T \left(A^{-1}\right)^T \bfg(\theta^*)\bigr] \nonumber \\
    &= \bfg(\theta^*)^T A^{-1} \E\bigl[\bfepsilon_\mathrm{shots} \bfepsilon_\mathrm{shots}^T\bigr] \left(A^{-1}\right)^T \bfg(\theta^*) \nonumber \\
    &= \frac{\sigma^2}{N_\mathrm{shots}}\, \bfg(\theta^*)^T A^{-1} \left(A^{-1}\right)^T \bfg(\theta^*) \nonumber \\
    &= \frac{\sigma^2}{N_\mathrm{shots}}\, \bfg(\theta^*)^T P\, \bfg(\theta^*),
    \label{eq:var_general}
\end{align}
where $P := \left(A^TA\right)^{-1} = A^{-1}\left(A^{-1}\right)^T$. Setting
$\bfg = \bfv,\, q = f^*$ yields \cref{eq:var_energy} and
$\bfg = \bfw,\, q = \theta^*$ yields \cref{eq:var_theta}.

\section{Expected value of variance using von Mises prior}
\label{app:expvar}
Throughout this appendix we use the Bessel-function ratios
\begin{align*}
    R(\kappa) &:= \frac{I_1(\kappa)}{I_0(\kappa)}, &
    R_2(\kappa) &:= \frac{I_2(\kappa)}{I_0(\kappa)},
\end{align*}
which arise as the trigonometric moments $\E[\cos(t(\theta^*-\mu))] = I_t(\kappa)/I_0(\kappa)$ of the von Mises distribution. We will further use the recurrence $I_{t-1}(\kappa) - I_{t+1}(\kappa) = \frac{2t}{\kappa}I_t(\kappa)$.
We can then use the characteristic function $\E_{\theta^*}\left[e^{it\theta^*} \right] =\frac{I_t(\kappa)}{I_0(\kappa)} e^{it\mu}$ to get the identities
\begin{align*}
    \E_{\theta^*}\left[\cos(\theta^*)\right] &= R(\kappa)\cos(\mu), \\
    \E_{\theta^*}\left[\sin(\theta^*)\right] &= R(\kappa)\sin(\mu), \\
    \E_{\theta^*}\left[\cos^2(\theta^*)\right] &= \frac{1}{2}(1 + R_2(\kappa)\cos(2\mu)), \\
    \E_{\theta^*}\left[\sin^2(\theta^*)\right] &= \frac{1}{2}(1 - R_2(\kappa)\cos(2\mu)), \\
    \E_{\theta^*}\left[\cos(\theta^*) \sin\theta^*\right] &= \frac{1}{2}R_2(\kappa)\sin(2\mu).
\end{align*}
\begin{proof}[Proof of \Cref{lemma:expvar}]
    Taking the expectation of \cref{eq:var_general} over $\theta^*$ and
    using
    $\bfg(\theta^*)^T P\, \bfg(\theta^*) = \Tr\bigl(P\, \bfg(\theta^*)\bfg(\theta^*)^T\bigr)$
    together with linearity of the trace gives
    \begin{align*}
        \E_{\theta^*}\bigl[\Var(\hat q)\bigr]
        &= \frac{\sigma^2}{N_\mathrm{shots}}
          \Tr\bigl(P\, \E_{\theta^*}\bigl[\bfg(\theta^*)\bfg(\theta^*)^T\bigr]\bigr) \\
        &= \frac{\sigma^2}{N_\mathrm{shots}} \Tr\bigl(P S_g\bigr).
    \end{align*}
\end{proof}

\subsection{\label{app:expvar_energy}Estimated minimum energy}
\begin{proof}[Proof of \Cref{thm:expvar_energy}]
The expected variance of the estimated minimum energy is then given as 
\begin{align*}
\E_{\theta^*}[\Var(\hat f^*)] &= \frac{\sigma^2}{N_\mathrm{shots}} \E_{\theta^*}\left[\bfv(\theta^*)^T P \bfv(\theta^*)\right] \\
&= \frac{\sigma^2}{N_\mathrm{shots}} \Tr\left(P \cdot S_v\right)
\end{align*}
where
\begin{align*}
    S_v &= \E_{\theta^*}[\bfv(\theta^*) \bfv(\theta^*)^T] \\
    &= \begin{bmatrix} 1 & \sqrt{2}\,R(\kappa)\cos\mu & \sqrt{2}\,R(\kappa)\sin\mu \\[2pt]
\sqrt{2}\,R(\kappa)\cos\mu & 1 + R_2(\kappa)\cos2\mu & R_2(\kappa)\sin2\mu \\[2pt]
\sqrt{2}\,R(\kappa)\sin\mu & R_2(\kappa)\sin2\mu & 1 - R_2(\kappa)\cos2\mu
\end{bmatrix}.
\end{align*}
For $\mu = \varphi$ the objective does not depend on $\varphi$. By angle addition, shifting the pivot rotates the $(b_2, b_3)$ columns of $A$ and the corresponding components of $\bfv(\theta^*)$ by the same angle. This rotation is orthogonal and conjugates both $P$ and $S_v$, so $\Tr\left(P S_v\right)$ is unchanged and depends only on the offset $\alpha$ and the concentration $\kappa$.
\end{proof}

For $\kappa = 0$ the prior is uniform over the circle, so $S_v = I$ and the optimal spacing is
\begin{equation*}
    \argmin_\alpha \Tr\left(P \cdot S_v\right)\Big|_{\kappa=0} = \frac{2\pi}{3},
\end{equation*}
recovering the equal-spacing result of \cite{endo_optimal_2023}.

\subsection{\label{app:expvar_theta}Estimated minimizer}
\begin{proof}[Proof of \Cref{thm:expvar_theta}]
We can use the same identities from above to find the expected variance of the estimated minimizer
\begin{align*}
\E_{\theta^*}[\Var(\hat \theta^*)] &= \frac{\sigma^2}{N_\mathrm{shots}} \E_{\theta^*}\left[\bfw(\theta^*)^T P \bfw(\theta^*)\right] \\
&= \frac{\sigma^2}{N_\mathrm{shots}} \Tr\left(P \cdot S_w\right)
\end{align*}
where
\begin{align*}
    S_w &= \E_{\theta^*}[\bfw(\theta^*) \bfw(\theta^*)^T] \\
    &= \frac{1}{2r^2}
    \begin{bmatrix}
        0 & 0 & 0 \\
        0 & 1 - R_2(\kappa)\cos(2\mu) & -R_2(\kappa)\sin(2\mu) \\
        0 & -R_2(\kappa)\sin(2\mu) & 1 + R_2(\kappa)\cos(2\mu)
    \end{bmatrix}.
\end{align*}
The $(b_2, b_3)$ components of $\bfw(\theta^*)$ rotate with the pivot in the same way, so the argument of \cref{app:expvar_energy}, that $\Tr\left(P S_w\right)$ does not depend on the pivot, applies again.
\end{proof}
As for \cref{app:expvar_energy}, for $\kappa = 0$ the prior is uniform over the circle, so $S_w = \frac{1}{2r^2}\diag\left(0,1,1\right)$ and the optimal spacing is
\begin{equation*}
    \argmin_\alpha \Tr \left(P \cdot S_w \right)\Big|_{\kappa=0} = \frac{2\pi}{3}.
\end{equation*}
\section{\label{app:alpha_closedform}Closed-form solution for optimal shift between observations}
\begin{proof}[Proof of \Cref{thm:optimal_shift}]
We seek the shift $\alpha$ between observations that minimizes the expected variance of the estimated minimizer,
\begin{equation*}
    \alpha^* = \argmin_\alpha \E_{\theta^*}\left[ \Var \left(\hat{\theta}^* \right)\right].
\end{equation*}
As stated in \cref{sec:vqe}, we assume $\frac{\sigma^2}{N_\mathrm{shots}}$ does not depend on $\alpha$ due to $\sigma$ being constant. Therefore we want to optimize
\begin{equation*}
    \alpha^* = \argmin_\alpha \Tr \left(P \cdot S_w\right).
\end{equation*}
From \cref{app:expvar_theta}, setting $\mu = 0$ since we always rescale around the pivot $\varphi = 0$, we can simplify $S_w$ and obtain
\begin{equation*}
    S_w = \frac{1}{2r^2}
    \begin{bmatrix}
        0 & 0 & 0 \\
        0 & 1 - R_2(\kappa) & 0 \\
        0 & 0 & 1 + R_2(\kappa)
    \end{bmatrix}.
\end{equation*}
Given $\bfPhi = (0,\alpha,-\alpha)^T$, the design matrix is
\begin{equation*}
    A = \begin{bmatrix}
        1 & \sqrt{2} & 0 \\
        1 & \sqrt{2}\cos\alpha & \sqrt{2}\sin\alpha \\
        1 & \sqrt{2}\cos\alpha & -\sqrt{2}\sin\alpha
    \end{bmatrix},
\end{equation*}
where we used $\cos(-\alpha) = \cos\alpha$ and $\sin(-\alpha) = -\sin\alpha$. The Gram matrix is
\begin{equation*}
    A^TA = \begin{bmatrix}
        3 & \sqrt{2}(1 + 2\cos\alpha) & 0 \\
        \sqrt{2}(1 + 2\cos\alpha) & 2 + 4\cos^2\alpha & 0 \\
        0 & 0 & 4\sin^2\alpha
    \end{bmatrix}.
\end{equation*}
Since $A^TA$ is block-diagonal with a $2\times 2$ block $B$ and a $1\times 1$ block $C = 4\sin^2\alpha$, we invert each block separately. The determinant of $B$ is
\begin{equation*}
    d = 3(2 + 4\cos^2\alpha) - 2(1 + 2\cos\alpha)^2 = 4(\cos\alpha-1)^2,
\end{equation*}
which gives
\begin{equation*}
    (A^TA)^{-1} = \begin{bmatrix}
        \frac{2(1+2\cos^2\alpha)}{d} & -\frac{\sqrt{2}(1 + 2\cos\alpha)}{d} & 0 \\[6pt]
        -\frac{\sqrt{2}(1 + 2\cos\alpha)}{d} & \frac{3}{d} & 0 \\[6pt]
        0 & 0 & \frac{1}{4\sin^2\alpha}
    \end{bmatrix}.
\end{equation*}
The full expression therefore reduces to
\begin{equation*}
    \Tr\left((A^TA)^{-1} \cdot S_w\right) 
    = \frac{1}{2r^2}\left(\frac{3(1 - R_2(\kappa))}{d} + \frac{1 + R_2(\kappa)}{4\sin^2\alpha} \right).
\end{equation*}
Dropping the constant factor $\frac{1}{2r^2}$, we define
\begin{equation*}
    h(\alpha) = \frac{3(1 - R_2(\kappa))}{4(\cos\alpha-1)^2} + \frac{1 + R_2(\kappa)}{4\sin^2\alpha}.
\end{equation*}
Setting $h'(\alpha) = 0$ yields
\begin{equation*}
    \frac{3(1 - R_2(\kappa))\sin\alpha}{2(1-\cos\alpha)^3} = -\frac{(1 + R_2(\kappa))\cos\alpha}{2\sin^3\alpha}.
\end{equation*}
We substitute $t = \tan(\alpha/2)$ using the identities
\begin{equation*}
    \sin\alpha = \frac{2t}{1+t^2}, \quad
    \cos\alpha = \frac{1-t^2}{1+t^2}, \quad
    1 - \cos\alpha = \frac{2t^2}{1+t^2}.
\end{equation*}
Inserting these into the stationarity condition and cancelling the common factor $(1+t^2)^2$ gives
\begin{equation*}
    \frac{3(1-R_2(\kappa))}{8t^5} = -\frac{(1+R_2(\kappa))(1-t^2)}{16t^3}.
\end{equation*}
Multiplying both sides by $16t^5$ and rearranging yields the quartic in $t$:
\begin{equation*}
    (1+R_2(\kappa))\,t^4 - (1 + R_2(\kappa))\,t^2 - 6(1 - R_2(\kappa)) = 0.
\end{equation*}
Setting $u=t^2$, the quadratic formula gives
\begin{equation*}
    u = \frac{1 \pm \eta}{2},
    \qquad
    \eta = \sqrt{\frac{25 - 23 R_2(\kappa)}{1 + R_2(\kappa)}}.
\end{equation*}
Since $R_2(\kappa) \in [0,1)$, we have $\eta \ge 1$, so the negative branch gives $u \le 0$ and is rejected ($u = t^2 \ge 0$). The positive branch yields
\begin{equation*}
    \alpha^* = 2\arctan\sqrt{\frac{1 + \eta}{2}}.
\end{equation*}
Both terms of $h$ are convex on $(0,\pi)$, with second derivatives
\begin{align*}
    \frac{d^2}{d\alpha^2}\frac{1}{(1-\cos\alpha)^2} &= \frac{2(3 + 2\cos\alpha)}{(1-\cos\alpha)^3}, \\
    \frac{d^2}{d\alpha^2}\frac{1}{\sin^2\alpha} &= \frac{2 + 4\cos^2\alpha}{\sin^4\alpha},
\end{align*}
both positive since $\cos\alpha > -1$ on $(0,\pi)$. As the coefficients $\frac{3}{4}(1 - R_2(\kappa))$ and $\frac{1}{4}(1 + R_2(\kappa))$ are nonnegative, $h$ is strictly convex and the stationary point above is its unique global minimizer.
\end{proof}
\section{\label{app:bound}Bound on the optimal shift angle}
\begin{proof}[Proof of \Cref{cor:bound}]
    Since $I_2(\kappa) < I_0(\kappa)$ for all $\kappa \ge 0$, we have $R_2(\kappa) \in [0,1)$. Writing $\eta = \sqrt{g(R_2)}$ with $g(R_2) = \frac{25-23R_2}{1+R_2}$, we have
    \begin{equation*}
        g'(R_2) = \frac{-23(1+R_2) - (25-23R_2)}{(1+R_2)^2} = \frac{-48}{(1+R_2)^2} < 0,
    \end{equation*}
    so $g$ is strictly decreasing on $[0,1)$ and, since $g > 0$ there, so is $\eta = \sqrt{g}$, which decreases strictly from $5$ toward $1$. Conversely, $\alpha^*(\eta) = 2\arctan\sqrt{(1+\eta)/2}$ is a composition of strictly increasing maps on $\eta \ge 1$ and is therefore strictly increasing. Hence \cref{eq:optimal_shift} maps $R_2 = 0$ to $\alpha^* = 2\pi/3$ and $R_2 \to 1$ to $\alpha^* \to \pi/2$, and strict monotonicity of both maps ensures the bound holds not only at these extremes but for every $\kappa$.
\end{proof}
\section{Derivation of the asymptotic minimizer variance}
\label{app:pivot_aligned}
\begin{proof}[Proof of \Cref{thm:asymptotic}]
We seek the variance of $\hat\theta^*$ in the limit $\varphi \to \theta^*$. Setting $\varphi = \theta^* = 0$ without loss of generality, \cref{eq:theta_opt} gives $(b_2, b_3) = (-r, 0)$, and \cref{eq:gradient_w} reduces to $\bfw = (0, 0, -1/r)^T$. Substituting into \cref{eq:var_theta} yields
\begin{equation*}
    \Var(\hat\theta^*) = \frac{\sigma^2}{r^2 N_\mathrm{shots}} (A^TA)^{-1}_{33} + O(N_\mathrm{shots}^{-2}).
\end{equation*}
With $\bfPhi = (0, \alpha, -\alpha)^T$ in \cref{eq:design_matrix} and using $\cos(-\alpha) = \cos\alpha$ and $\sin(-\alpha) = -\sin\alpha$, the Gram matrix is
\begin{equation*}
    A^TA = \begin{bmatrix}
        3 & \sqrt{2}(1 + 2\cos\alpha) & 0 \\
        \sqrt{2}(1 + 2\cos\alpha) & 2 + 4\cos^2\alpha & 0 \\
        0 & 0 & 4\sin^2\alpha
    \end{bmatrix}.
\end{equation*}
Since $A^TA$ is block-diagonal, the $(3,3)$ entry of its inverse is obtained from the $1\times 1$ block alone,
\begin{equation*}
    (A^TA)^{-1}_{33} = \frac{1}{4\sin^2\alpha},
\end{equation*}
which yields the closed form stated in \cref{thm:asymptotic}.
\end{proof}

For a general pivot $\varphi$, the antisymmetric difference of the two symmetric observations is obtained by applying the sum-to-product identities to \cref{eq:sinusoidal},
\begin{equation*}
    f(\varphi + \alpha) - f(\varphi - \alpha) = 2\sqrt{2}\sin(\alpha) \bigl[b_3 \cos(\varphi) - b_2 \sin(\varphi)\bigr].
\end{equation*}
Setting $\varphi = \theta^* = 0$ gives $\cos(\varphi) = 1$, $\sin(\varphi) = 0$ and $(b_2, b_3) = (-r, 0)$, so the bracketed term reduces to $b_3$, recovering the aligned-case expression. This confirms that the asymptotic form derived above is the relevant limit for iterative sequential optimization, where $\varphi$ is updated toward $\theta^*$ at every step.
\section{Experimental setup}

\subsection{Circuit ansätze}
\label{app:circuits}

The variational ansätze used in our experiments are shown in Figures~\ref{fig:circuit_esu2}--\ref{fig:circuit_hea}.

\begin{figure*}[t]
    \centering
    \includegraphics[width=\textwidth]{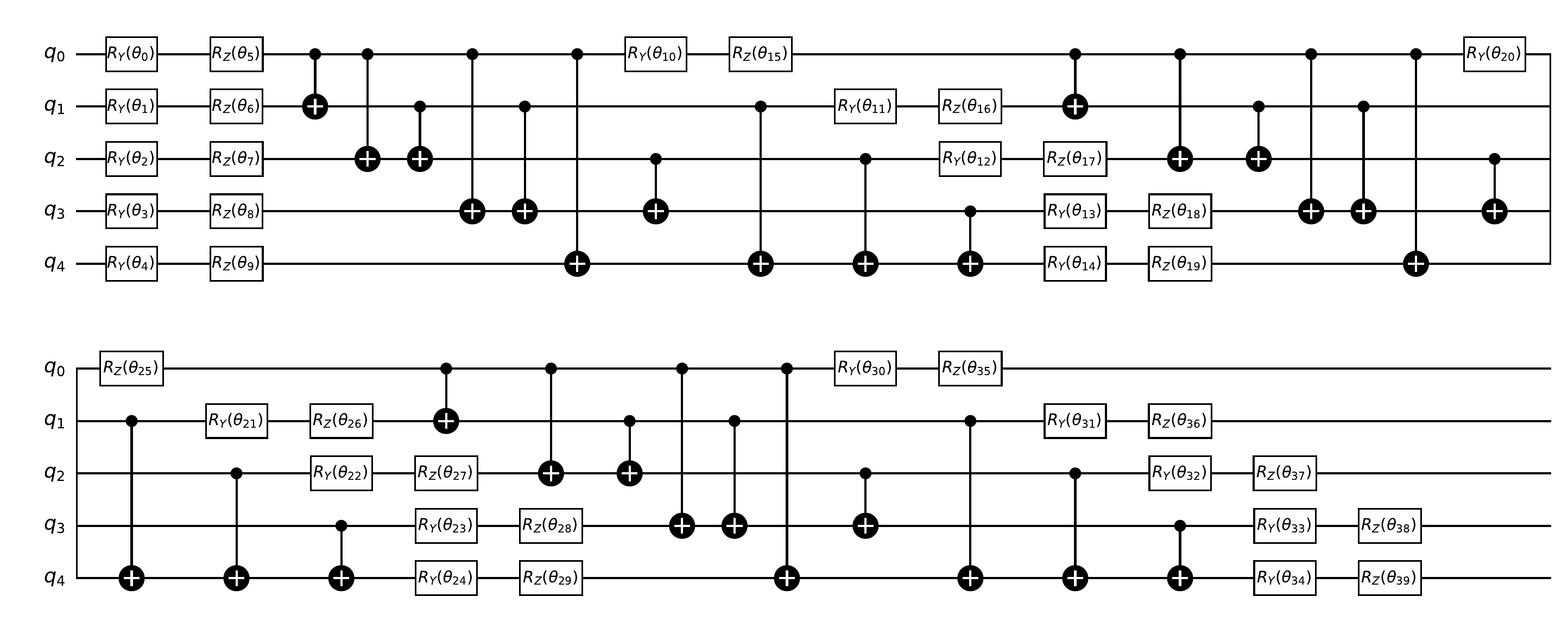}
    \caption{The efficient SU(2) ansatz with $Q=5$ qubits and $L=3$ layers.}
    \label{fig:circuit_esu2}
\end{figure*}

\begin{figure*}[t]
    \centering
    \includegraphics[width=\textwidth]{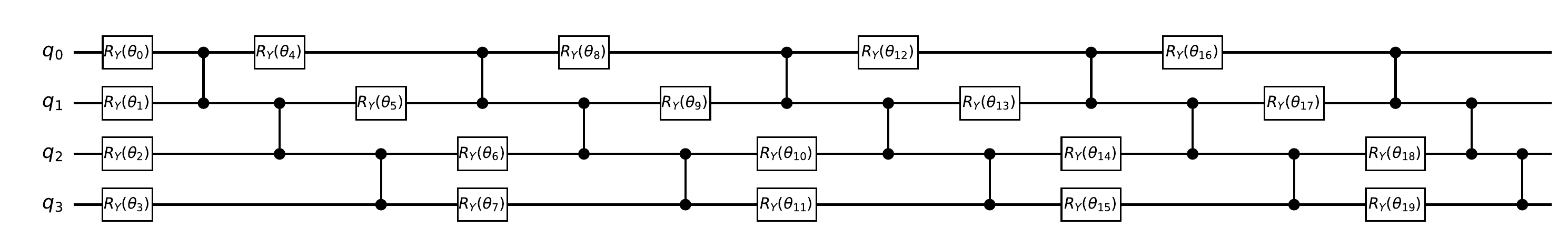}
    \caption{The hardware-efficient ansatz (HEA) with $Q=4$ qubits and $L=5$ layers.}
    \label{fig:circuit_hea}
\end{figure*}

\subsection{Linear interpolation of shift angles}
\label{app:shift_linear}

\begin{figure*}[t]
    \centering
    \includegraphics[width=0.9\linewidth]{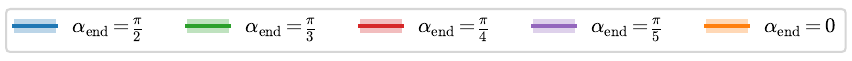}
    \includegraphics[width=\linewidth]{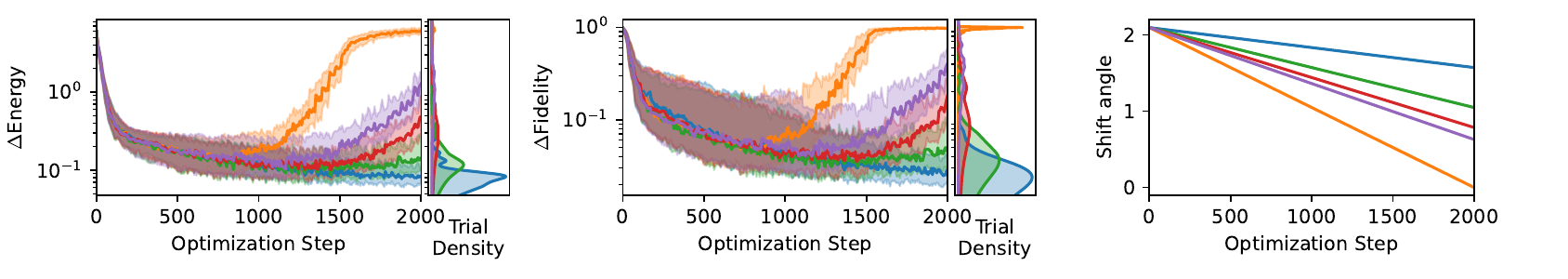}
    \caption{Rotosolve VQE for the TFIM on a $Q=5$-qubit, $L=3$-layer ESU2 ansatz; $1000$ shots per evaluation, $100$ trials of $2000$ steps. Each colour fixes a terminal shift angle $\alpha_\mathrm{end}$, with the measurement spacing linearly contracted from the equidistant $\alpha_0=\frac{2\pi}{3}$ to $\alpha_\mathrm{end}$ as given by \cref{eq:linear_schedule} (bottom panel). \trialdesc{} Too small $\alpha_\mathrm{end}$ destabilizes optimization, with $\alpha_\mathrm{end}=0$ collapsing entirely.}
    \label{fig:shift_schedule}
\end{figure*}

To test how far the measurement spacing can be contracted, we linearly interpolate $\alpha$ from the equidistant value $\alpha_0 = 2\pi/3$ at step $0$ to a target $\alpha_\mathrm{end}$ at the final optimization step $T$,

\begin{equation}
  \alpha(t) = \Bigl(1-\frac{t}{T}\Bigr)\frac{2\pi}{3} + \frac{t}{T}\,\alpha_\mathrm{end},
  \quad t = 0,\dots,T ,
  \label{eq:linear_schedule}
\end{equation}
and sweep $\alpha_\mathrm{end} \in \{\pi/2,\,\pi/3,\,\pi/4,\,\pi/5,\,0\}$ (Fig.~\ref{fig:shift_schedule}, bottom).

All schedules track each other while $\alpha$ stays near $\alpha_0$, then separate in order of $\alpha_\mathrm{end}$. Consistent with \cref{thm:asymptotic}, contracting below $\frac{\pi}{2}$ inflates the per-step minimizer variance. Only $\alpha_\mathrm{end}=\frac{\pi}{2}$ remains stable through the full run, while every smaller target eventually degrades, with outright collapse at $\alpha_\mathrm{end}=0$. This confirms the bound at the full-optimization level, the spacing should not be contracted below $\frac{\pi}{2}$ when the per-iteration objective is the minimizer $\hat\theta^*$.

\subsection{Fixed shift angles}
\label{app:shift_fixed}
\begin{figure*}[t]
    \centering
    \includegraphics[width=0.9\linewidth]{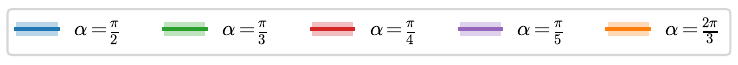}
    \includegraphics[width=\linewidth]{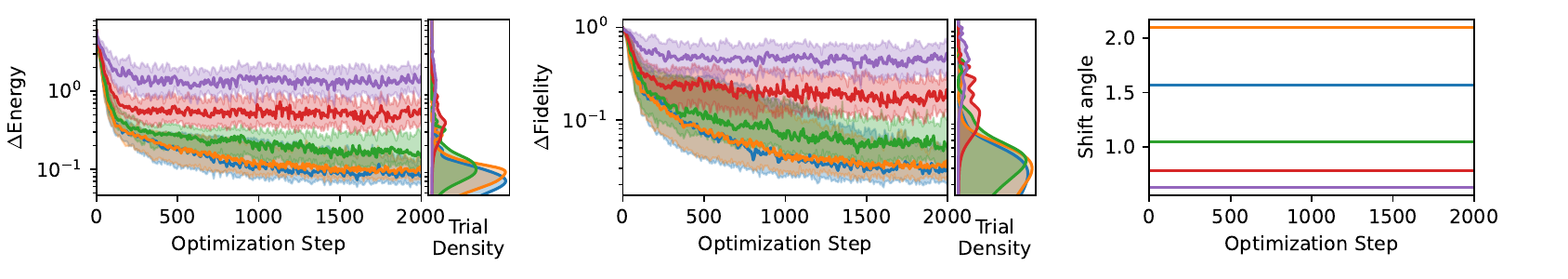}
    \caption{Rotosolve VQE for the TFIM on a $Q=5$-qubit, $L=3$-layer ESU2 ansatz; $1000$ shots per evaluation, $100$ trials of $2000$ steps. Each colour fixes a constant shift angle $\alpha$ throughout optimization, ranging from the equidistant $\alpha=\frac{2\pi}{3}$ down to $\alpha=\frac{\pi}{5}$. \trialdesc}
    \label{fig:shift_fixed}
\end{figure*}

Second, we evaluate fixed shifts below $\frac{\pi}{2}$, running default Rotosolve with constant $\alpha \in \{\frac{2\pi}{3},\,\frac{\pi}{2},\,\frac{\pi}{3},\,\frac{\pi}{4},\,\frac{\pi}{5}\}$. As shown in \Cref{fig:shift_fixed}, the two shifts within the predicted optimal range, $\alpha=\frac{2\pi}{3}$ and $\alpha=\frac{\pi}{2}$, perform comparably throughout. Consistent with \cref{thm:asymptotic}, every fixed shift below $\frac{\pi}{2}$ drastically degrades performance, in order of contraction. Unlike the interpolated schedules, the variance inflation acts from step $0$, so the shifts separate immediately. This supports keeping $\alpha \in [\frac{\pi}{2}, \frac{2\pi}{3}]$ throughout as there is no benefit to contracting below $\frac{\pi}{2}$ at any stage.

\clearpage
\bibliography{references}

@article{acharya_quantum_2025,
  title = {Quantum Error Correction below the Surface Code Threshold},
  author = {{Google Quantum AI and Collaborators}},
  year = 2025,
  month = feb,
  journal = {Nature},
  volume = {638},
  number = {8052},
  pages = {920--926},
  publisher = {Nature Publishing Group},
  issn = {1476-4687},
  doi = {10.1038/s41586-024-08449-y},
  urldate = {2026-06-10},
  copyright = {2024 The Author(s)},
  langid = {english}
}

@inproceedings{anders_adaptive_2024,
  title = {Adaptive {{Observation Cost Control}} for {{Variational Quantum Eigensolvers}}},
  booktitle = {Proceedings of the 41st {{International Conference}} on {{Machine Learning}}},
  author = {Anders, Christopher J. and Nicoli, Kim Andrea and Wu, Bingting and Elosegui, Naima and Pedrielli, Samuele and Funcke, Lena and Jansen, Karl and K{\"u}hn, Stefan and Nakajima, Shinichi},
  year = 2024,
  month = jul,
  pages = {1557--1578},
  publisher = {PMLR},
  issn = {2640-3498},
  urldate = {2026-06-02},
  langid = {english}
}

@article{biamonte_quantum_2017,
  title = {Quantum Machine Learning},
  author = {Biamonte, Jacob and Wittek, Peter and Pancotti, Nicola and Rebentrost, Patrick and Wiebe, Nathan and Lloyd, Seth},
  year = 2017,
  month = sep,
  journal = {Nature},
  volume = {549},
  number = {7671},
  pages = {195--202},
  publisher = {Nature Publishing Group},
  issn = {1476-4687},
  doi = {10.1038/nature23474},
  urldate = {2025-11-25},
  copyright = {2017 Macmillan Publishers Limited, part of Springer Nature. All rights reserved.},
  langid = {english}
}

@article{caditazi_folded_2024,
  title = {Folded {{Spectrum VQE}}: {{A Quantum Computing Method}} for the {{Calculation}} of {{Molecular Excited States}}},
  shorttitle = {Folded {{Spectrum VQE}}},
  author = {Cadi Tazi, Lila and Thom, Alex J. W.},
  year = 2024,
  month = mar,
  journal = {Journal of Chemical Theory and Computation},
  volume = {20},
  number = {6},
  pages = {2491--2504},
  issn = {1549-9618, 1549-9626},
  doi = {10.1021/acs.jctc.3c01378},
  urldate = {2026-06-10},
  copyright = {https://creativecommons.org/licenses/by/4.0/},
  langid = {english}
}

@article{cai_quantum_2023,
  title = {Quantum Error Mitigation},
  author = {Cai, Zhenyu and Babbush, Ryan and Benjamin, Simon C. and Endo, Suguru and Huggins, William J. and Li, Ying and McClean, Jarrod R. and O'Brien, Thomas E.},
  year = 2023,
  month = dec,
  journal = {Reviews of Modern Physics},
  volume = {95},
  number = {4},
  pages = {045005},
  issn = {0034-6861, 1539-0756},
  doi = {10.1103/RevModPhys.95.045005},
  urldate = {2026-02-03},
  langid = {english}
}

@article{endo_optimal_2023,
  title = {Optimal Parameter Configurations for Sequential Optimization of the Variational Quantum Eigensolver},
  author = {Endo, Katsuhiro and Sato, Yuki and Raymond, Rudy and Wada, Kaito and Yamamoto, Naoki and Watanabe, Hiroshi C.},
  year = 2023,
  month = nov,
  journal = {Physical Review Research},
  volume = {5},
  number = {4},
  pages = {043136},
  issn = {2643-1564},
  doi = {10.1103/PhysRevResearch.5.043136},
  urldate = {2025-10-20},
  langid = {english}
}

@article{fauseweh_quantum_2024,
  title = {Quantum Many-Body Simulations on Digital Quantum Computers: {{State-of-the-art}} and Future Challenges},
  shorttitle = {Quantum Many-Body Simulations on Digital Quantum Computers},
  author = {Fauseweh, Benedikt},
  year = 2024,
  month = mar,
  journal = {Nature Communications},
  volume = {15},
  number = {1},
  pages = {2123},
  issn = {2041-1723},
  doi = {10.1038/s41467-024-46402-9},
  urldate = {2026-06-10},
  langid = {english}
}

@article{hillen_moments_2017,
  title = {Moments of von {{Mises}} and {{Fisher}} Distributions and Applications},
  author = {Hillen, Thomas and Painter, Kevin J and Swan, Amanda C. and Murtha, Albert D.},
  year = 2017,
  month = jun,
  journal = {Mathematical Biosciences and Engineering},
  volume = {14},
  number = {3},
  pages = {673--694},
  issn = {1551-0018},
  doi = {10.3934/mbe.2017038}
}

@misc{javadi-abhari_quantum_2024,
  title = {Quantum Computing with {{Qiskit}}},
  author = {{Javadi-Abhari}, Ali and Treinish, Matthew and Krsulich, Kevin and Wood, Christopher J. and Lishman, Jake and Gacon, Julien and Martiel, Simon and Nation, Paul D. and Bishop, Lev S. and Cross, Andrew W. and Johnson, Blake R. and Gambetta, Jay M.},
  year = 2024,
  month = jun,
  number = {arXiv:2405.08810},
  eprint = {2405.08810},
  primaryclass = {quant-ph},
  publisher = {arXiv},
  doi = {10.48550/arXiv.2405.08810},
  urldate = {2026-03-24},
  archiveprefix = {arXiv}
}

@article{jiang_error_2024,
  title = {Error Mitigation in Variational Quantum Eigensolvers Using Tailored Probabilistic Machine Learning},
  author = {Jiang, Tao and Rogers, John and Frank, Marius S. and Christiansen, Ove and Yao, Yong-Xin and Lanat{\`a}, Nicola},
  year = 2024,
  month = jul,
  journal = {Physical Review Research},
  volume = {6},
  number = {3},
  pages = {033069},
  issn = {2643-1564},
  doi = {10.1103/PhysRevResearch.6.033069},
  urldate = {2026-02-06},
  langid = {english}
}

@article{lai_interpolationbased_2026,
  title = {Interpolation-Based Coordinate Descent Method for Parameterized Quantum Circuits},
  author = {Lai, Zhijian and Hu, Jiang and Ko, Taehee and Wu, Jiayuan and An, Dong},
  year = 2026,
  month = jan,
  journal = {Communications Physics},
  volume = {9},
  number = {1},
  pages = {41},
  publisher = {Nature Publishing Group},
  issn = {2399-3650},
  doi = {10.1038/s42005-025-02473-8},
  urldate = {2026-02-17},
  copyright = {2026 The Author(s)},
  langid = {english}
}

@article{mcclean_theory_2016,
  title = {The Theory of Variational Hybrid Quantum-Classical Algorithms},
  author = {McClean, Jarrod R and Romero, Jonathan and Babbush, Ryan and {Aspuru-Guzik}, Al{\'a}n},
  year = 2016,
  month = feb,
  journal = {New Journal of Physics},
  volume = {18},
  number = {2},
  pages = {023023},
  publisher = {IOP Publishing},
  issn = {1367-2630},
  doi = {10.1088/1367-2630/18/2/023023},
  urldate = {2026-02-10},
  langid = {english}
}

@article{meth_simulating_2025,
  title = {Simulating Two-Dimensional Lattice Gauge Theories on a Qudit Quantum Computer},
  author = {Meth, Michael and Zhang, Jinglei and Haase, Jan F. and Edmunds, Claire and Postler, Lukas and Jena, Andrew J. and Steiner, Alex and Dellantonio, Luca and Blatt, Rainer and Zoller, Peter and Monz, Thomas and Schindler, Philipp and Muschik, Christine and Ringbauer, Martin},
  year = 2025,
  month = apr,
  journal = {Nature Physics},
  volume = {21},
  number = {4},
  pages = {570--576},
  publisher = {Nature Publishing Group},
  issn = {1745-2481},
  doi = {10.1038/s41567-025-02797-w},
  urldate = {2025-11-24},
  copyright = {2025 The Author(s)},
  langid = {english}
}

@article{moll_quantum_2018,
  title = {Quantum Optimization Using Variational Algorithms on Near-Term Quantum Devices},
  author = {Moll, Nikolaj and Barkoutsos, Panagiotis and Bishop, Lev S and Chow, Jerry M and Cross, Andrew and Egger, Daniel J and Filipp, Stefan and Fuhrer, Andreas and Gambetta, Jay M and Ganzhorn, Marc and Kandala, Abhinav and Mezzacapo, Antonio and M{\"u}ller, Peter and Riess, Walter and Salis, Gian and Smolin, John and Tavernelli, Ivano and Temme, Kristan},
  year = 2018,
  month = jun,
  journal = {Quantum Science and Technology},
  volume = {3},
  number = {3},
  pages = {030503},
  publisher = {IOP Publishing},
  issn = {2058-9565},
  doi = {10.1088/2058-9565/aab822},
  urldate = {2025-11-25},
  langid = {english}
}

@article{nakanishi_sequential_2020,
  title = {Sequential Minimal Optimization for Quantum-Classical Hybrid Algorithms},
  author = {Nakanishi, Ken M. and Fujii, Keisuke and Todo, Synge},
  year = 2020,
  month = oct,
  journal = {Physical Review Research},
  volume = {2},
  number = {4},
  eprint = {1903.12166},
  primaryclass = {quant-ph},
  pages = {043158},
  issn = {2643-1564},
  doi = {10.1103/PhysRevResearch.2.043158},
  urldate = {2025-09-12},
  archiveprefix = {arXiv}
}

@article{nicoli_physicsinformed_2023,
  title = {Physics-{{Informed Bayesian Optimization}} of {{Variational Quantum Circuits}}},
  author = {Nicoli, Kim and Anders, Christopher J. and Funcke, Lena and Hartung, Tobias and Jansen, Karl and K{\"u}hn, Stefan and M{\"u}ller, Klaus-Robert and Stornati, Paolo and Kessel, Pan and Nakajima, Shinichi},
  year = 2023,
  month = dec,
  journal = {Advances in Neural Information Processing Systems},
  volume = {36},
  pages = {18341--18376},
  urldate = {2026-06-02},
  langid = {english}
}

@article{ostaszewski_structure_2021,
  title = {Structure Optimization for Parameterized Quantum Circuits},
  author = {Ostaszewski, Mateusz and Grant, Edward and Benedetti, Marcello},
  year = 2021,
  month = jan,
  journal = {Quantum},
  volume = {5},
  eprint = {1905.09692},
  primaryclass = {quant-ph},
  pages = {391},
  issn = {2521-327X},
  doi = {10.22331/q-2021-01-28-391},
  urldate = {2025-10-20},
  archiveprefix = {arXiv}
}

@inproceedings{pedrielli_bayesian_2025,
  title = {Bayesian {{Parameter Shift Rules}} in {{Variational Quantum Eigensolvers}}},
  booktitle = {The {{Fourteenth International Conference}} on {{Learning Representations}}},
  author = {Pedrielli, Samuele and Anders, Christopher J. and Funcke, Lena and Jansen, Karl and Nicoli, Kim Andrea and Nakajima, Shinichi},
  year = 2025,
  month = oct,
  urldate = {2026-06-09},
  langid = {english}
}

@misc{pedrielli_bias_2026,
  title = {Bias {{Analysis}} and {{Regularization}} of {{Sequential Minimal Optimization}} in {{Variational Quantum Eigensolvers}}},
  author = {Pedrielli, Samuele and Stalschus, Frederik and K{\"u}hn, Stefan and Jansen, Karl and Nicoli, Kim A. and Nakajima, Shinichi},
  year = 2026,
  month = may,
  number = {arXiv:2605.15813},
  eprint = {2605.15813},
  primaryclass = {quant-ph},
  publisher = {arXiv},
  doi = {10.48550/arXiv.2605.15813},
  urldate = {2026-06-09},
  archiveprefix = {arXiv}
}

@article{peruzzo_variational_2014,
  title = {A Variational Eigenvalue Solver on a Photonic Quantum Processor},
  author = {Peruzzo, Alberto and McClean, Jarrod and Shadbolt, Peter and Yung, Man-Hong and Zhou, Xiao-Qi and Love, Peter J. and {Aspuru-Guzik}, Al{\'a}n and O'Brien, Jeremy L.},
  year = 2014,
  month = jul,
  journal = {Nature Communications},
  volume = {5},
  pages = {4213},
  issn = {2041-1723},
  doi = {10.1038/ncomms5213},
  urldate = {2025-11-19},
  pmcid = {PMC4124861},
  pmid = {25055053}
}

@techreport{platt_sequential_1998,
  title = {Sequential {{Minimal Optimization}}: {{A Fast Algorithm}} for {{Training Support Vector Machines}}},
  author = {Platt, John C},
  year = 1998,
  month = apr,
  institution = {Microsoft},
  langid = {english}
}

@article{smith_simulating_2019,
  title = {Simulating Quantum Many-Body Dynamics on a Current Digital Quantum Computer},
  author = {Smith, Adam and Kim, M. S. and Pollmann, Frank and Knolle, Johannes},
  year = 2019,
  month = nov,
  journal = {npj Quantum Information},
  volume = {5},
  number = {1},
  pages = {106},
  publisher = {Nature Publishing Group},
  issn = {2056-6387},
  doi = {10.1038/s41534-019-0217-0},
  urldate = {2025-11-24},
  copyright = {2019 The Author(s)},
  langid = {english}
}

@article{sung_using_2020,
  title = {Using Models to Improve Optimizers for Variational Quantum Algorithms},
  author = {Sung, Kevin J and Yao, Jiahao and Harrigan, Matthew P and Rubin, Nicholas C and Jiang, Zhang and Lin, Lin and Babbush, Ryan and McClean, Jarrod R},
  year = 2020,
  month = sep,
  journal = {Quantum Science and Technology},
  volume = {5},
  number = {4},
  pages = {044008},
  publisher = {IOP Publishing},
  issn = {2058-9565},
  doi = {10.1088/2058-9565/abb6d9},
  urldate = {2026-05-18},
  langid = {english}
}

@article{tamiya_stochastic_2022,
  title = {Stochastic Gradient Line {{Bayesian}} Optimization for Efficient Noise-Robust Optimization of Parameterized Quantum Circuits},
  author = {Tamiya, Shiro and Yamasaki, Hayata},
  year = 2022,
  month = jul,
  journal = {npj Quantum Information},
  volume = {8},
  number = {1},
  pages = {90},
  issn = {2056-6387},
  doi = {10.1038/s41534-022-00592-6},
  urldate = {2025-09-18},
  langid = {english}
}

@article{zhang_variational_2022,
  title = {Variational Quantum Eigensolvers by Variance Minimization},
  author = {Zhang, Dan-Bo and Chen, Bin-Lin and Yuan, Zhan-Hao and Yin, Tao},
  year = 2022,
  month = nov,
  journal = {Chinese Physics B},
  volume = {31},
  number = {12},
  pages = {120301},
  publisher = {{Chinese Physical Society and IOP Publishing Ltd}},
  issn = {1674-1056},
  doi = {10.1088/1674-1056/ac8a8d},
  urldate = {2026-02-10},
  langid = {english}
}

\end{document}